\documentclass{llncs}
\usepackage{amssymb}

\begin{document}

\title{Scope Logic: Extending Hoare Logic for Pointer Program Verification
\thanks{This paper is supported by the Chinese HighTech project, grant no. 2009AA01Z148; and the HGJ project, grant no. 2009ZX01036-001-001}}
\author{ZHAO Jianhua, LI Xuandong}

\institute{State Key Laboratory of Novel Software Technology\\
    Dept. of Computer Sci. and Tech. Nanjing University\\
    Nanjing, Jiangsu, P.R.China 210093\\
    zhaojh@nju.edu.cn}
\maketitle

\newcommand{\seml}{[\![}
\newcommand{\semr}{]\!]}

\begin{abstract}
This paper presents an extension to Hoare logic for pointer program verification.
First, the Logic for Partial Function (LPF) used by VDM is extended to specify memory access using pointers and memory layout of composite types.
Then, the concepts of data-retrieve functions ( DRF ) and memory-scope functions (MSF) are introduced in this paper. People can define DRFs to retrieve abstract values from interconnected concrete  data objects. The definition of the corresponding MSF of a DRF can be derived syntactically  from the definition of the DRF. This MSF computes the set of memory units accessed when the DRF retrieves an abstract value. This memory unit set is called the memory scope of the abstract value. Finally, the proof rule of assignment statements in Hoare's logic is modified to deal with pointers. The basic idea is that a virtual value keeps unmodified as long as no memory unit in its scope is over-written. Another proof rule is added for memory allocation statements. The consequence rule and the rules for control-flow statements are slightly modified. They are essentially same as their original version in Hoare logic.

An example is presented to show the efficacy of this logic. We also give some heuristics on how to verify pointer programs.
\end{abstract}

\section{Introduction}
To reasoning the correctness of programs, C.A.R. Hoare presented an axiomatic system for specifying and verifying programs\cite{HOARE-1}\cite{HOARE-2}.
However, this logic can not deal with pointer programs because of pointer alias, i.e. many pointers may refer to the same location.
A few extensions to Hoare logic have been made to deal with pointers or shared mutable data structures \cite{RODNEY}\cite{STEPHEN}\cite{JOSEPH}.
Among them, separation logic \cite{SEPLOG} is one of the most successful extensions.
That logic uses a memory model which consists of two parts: the stack and the heap. Pointers can only refer to data objects in the heap. Separation logic extends
the predicate calculus with the separation operator, which can separate the heap into different disjoint parts.
Then the Hoare logic is extended with a set of proof rules for heap lookup, heap mutation and variable assignment. Though a few programs have been used to demonstrate the
potential of local reasoning for scalability\cite{SEPEXAMPLE}, verifying programs using separation logic is still very difficult.

This paper presents an extension to Hoare logic for verification of pointer programs. This logic uses an extension of the Logic for Partial Functions (LPF) \cite{LPF} to describe pre- and post-conditions of code fragments. Three type constructors are introduced to construct composite types and pointer types used in programs. Several kinds of function symbols associated with these types, together with a set of proof rules, are introduced to model and specify the memory layout/access in pointer programs.

In this logic, people can define recursive functions to retrieve abstract values from interconnected concrete data objects. These functions are called data-retrieve functions (DRFs). DRFs are recursively defined based on basic function symbols and the memory access/layout function symbols.
For each DRF $f$, there is a memory-scope function $\mathfrak{M}(f)$ of which the definition can be constructed syntactically from the definition of $f$. If an application of $f$ results in an abstract value, then an application of $\mathfrak{M}(f)$ to same arguments results in the set of memory units accessed during the application of $f$. During program executions, the application of $f$ to same arguments results in same abstract value as long as no memory unit in this set is modified.

In this logic, program specifications are of the form $\mathbb{P}\vdash p\{c\}r$, where $\mathbb{P}$ is a set of LPF formulae (usually a set of function definitions), $c$ is a fragment of code, and $q,r$ are the pre-condition and post-condition respectively. Such a specification means that if all the formulae in $\mathbb{P}$ hold for arbitrary program states, and $c$ starts its execution from a program state satisfying $q$; then the state must satisfy $r$ when $c$ stops.

This paper is organized as follows. An extension to LPF is presented in Section~\ref{SEC-EXT-LPF}. To model memory access and layout in pointer programs, several kinds of new function symbols and constants are introduced into LPF. A set of proof rules are introduced to specify these function symbols and constants. In Section~\ref{SUBSEC-SCOPE}, the concept `memory scope forms' of terms and `memory scope functions' (MSFs) are introduced. A proof rule is introduced to specify how definitions of MSFs can be constructed. A property about memory scope forms is also given in this section. The syntax of a small program language is given in Section~\ref{SEC-PROG-SYNTAX}. The semantic of this program language is also briefly described in this section. The syntax and meaning of program specifications are given in Section~\ref{SEC-SYN-SPEC}. The extension to Hoare logic is presented in Section~\ref{SEC-AXIOM}.
The proof rule for assignment statements is modified to dealing with the pointer alias problem. Another proof rule is introduced for memory allocation statements.
Section~\ref{SEC-PROOF} presents a formal verification of the  running example in this paper.  Section~\ref{SEC-HEU} gives some heuristics on program verifications using our logic.  Section~\ref{SEC-CONCLUSION} concludes this paper.

In Appendix~\ref{APP-BST-ADDNODE}, we verify another program which inserts a new node to a binary search tree. In Appendix~\ref{APP-ABS}, we use a simplified version of the Schorre-Waite algorithm
 to show that our logic can help people think about program verification in different abstract levels.

\subsection{Preliminary of the logic for partial functions}
The logic for partial functions (LPF) used in Vienna Development Method (VDM) can reason about undefinedness, (abstract) types, and recursive partial function definitions.
The syntax of LPF terms and formulae is briefly described here.  A term of LPF can be one of the following forms:
\begin{enumerate}
\item a variable symbol;
\item $f(e_1,\dots, e_n)$ if $f$ is a function symbol, $arity(f)=n$ and $e_1,\dots, e_n$ are terms,
\item $p\,?\,e_1:e_2$, where $p$ is a formula.
\end{enumerate}
A formula of LPF can be one of the following forms:
\begin{enumerate}
\item a boolean-typed term,
\item $\circledast$; ($\circledast$ denotes the neither-true-nor-false value. It is originally represented by the symbol $\ast$ in LPF papers, but $\ast$ is used to denote the memory access function in this paper.)
\item $P(e_1,\dots, e_n)$, if $P$ is a predicate symbol and $arity(P)=n$, and $e_1,\dots,e_n$ are terms. In this paper, we view a predicate symbol as a \textbf{boolean}-typed function symbol.
\item $e_1=e_2$,  where $e_1,e_2$ are terms,
\item $e:t$, where $e$ is a term and $t$ is a type symbol.
\item $\Delta A$, $\neg A$, $A_1\land  A_2$ are formulae if $A,A_1,A_2$ are formulae.
\item $\forall x:t\cdot A$, where $x$ is a variable, $t$ is a type symbol, and $A$ is a formula.
\item $f(x_1:T_1,\dots, x_n:T_n)\triangleq e$, where $e$ is a term,  and all the free variables in $t$ are in the set $\{x_1,\dots,x_n\}$.
\end{enumerate}
For the proof rules, semantics and other detail information of LPF, we refer readers to \cite{LPF}.

The LPF formulae used in our logic have a constraint: the logical connectives, $\circledast$, $\Delta$ and quantifiers can not occur in a term. Specifically, in a conditional form $p?e_1:e_2$, $p$ contains no logical connective and quantifier. However, we can use some operators like \textbf{cand}, \textbf{cor}, $\dots$, in terms. These operators can be defined using conditional forms. This constraint makes it possible to define the memory scope form of terms.

\section{The extension of the logic for partial functions}\label{SEC-EXT-LPF}
 In this paper, LPF is extended to deal with issues about memory access/layout, composite and pointer program types, data-retrieve functions and memory-scope functions. Now we first extend LPF with program types and associated function symbols.

\subsection{Program types and associated function symbols}\label{SUB-SEC-NEWFUNC}
In LPF, a type can be either a basic type such as $\mathbf{integer}$ and $\mathbf{boolean}$, or a type constructed using type constructors such as
\textbf{SetOf}, \textbf{SeqOf} and \textbf{Map}. However, the abstract types constructed using these type constructors can not be used directly in imperative programs.
To deal with types appeared in programs, we introduce three new type constructors into LPF: pointer (\textbf{P}), array (\textbf{ARR}), and record (\textbf{REC}). We call the types that can appear in programs as \emph{P-types}.
\begin{enumerate}
\item $\mathbf{integer}$ and $\mathbf{boolean}$ are P-types;
\item Let $t,t_1,\dots, t_k$ be P-types, $n_1,n_2,\dots,n_k$ are $k$ different names, $c$ is an positive integer constant.
    $\textbf{P}(t)$, $\textbf{ARR}(t, c)$, and $\textbf{REC}((n_1,t_1)\times \dots\times (n_k, t_k))$ are also P-types.
\end{enumerate}
We allow a record type $t$ has one or more fields with type $\textbf{P}(t)$ such that we can deal with recursive data types in our program language.
We use $\textbf{Ptr}$ as the super type of all pointer types $\textbf{P}(t)$, where $t$ is a P-type. The abstract type constructors $\textbf{Map}, \textbf{SetOf}, \textbf{SeqOf}$ can not be applied to composite program types. However, these type constructors can be applied to pointer types to form new abstract types. That is, we can get an abstract $\textbf{SetOf}(\textbf{P}(t))$ for some P-type $t$, but can not get an abstract type $\textbf{SetOf}(\textbf{Rec}((n_1,t_1)\times \dots\times (n_k, t_k))$.

The following constant and function symbols associated with P-types are introduced.
\begin{enumerate}
\item A program can declare a finite set of program variables with P-types. For each program variable $v$ declared with P-type $t$, $\& v$ is a constant with type $\textbf{P}(t)$.
\item For each pointer type $t$, there is a $t$-typed constant $\textbf{nil}_t$. The type subscript $t$ can be omitted if there is no ambiguity caused.
\item  A partial function $\ast:\textbf{Ptr}\rightarrow \textbf{Ptr}\cup \textbf{integer}\cup \textbf{boolean}$. We write an application of $\ast$ to $e$ as $\ast e$.
For a non-nil pointer $r$ with type $\textbf{P}(t)$, where $t$ is \textbf{integer}, \textbf{boolean} or a pointer type, $\ast r$ is a $t$-typed value. An application of this function symbol models a memory unit access.
\item For each array type $t=\textbf{ARR}(t',c)$, there is a partial function $\&[]_t:\textbf{P}(t)\times \textbf{integer}\rightarrow \textbf{P}(t')$. We write an application of such function as $\&e[i]_t$ instead of $\&[]_t(e,i)$. The type subscripts can be omitted if there is no ambiguity caused. These function symbols model the memory layout of array types. Intuitively speaking, if $e$ is a non-nil reference to a $t$-typed data object, $\&e[i]$ is the reference to the $i$th element. $\&e[i]$ is defined if and only if $e\neq \textbf{nil}$ and $0\le i<c$.
\item For each record type $t=\textbf{REC}((n_1,t_1)\times \dots\times (n_k, t_k))$ and a name $n_i$ $(1\le i\le k)$, we have a partial function $\&\!\!\rightarrow_{t}\!\!n_i:\textbf{P}(t)\rightarrow \textbf{P}(t_i)$. It is only undefined on the constant $\textbf{nil}_t$. We write an application of this function symbol to $e$ as $\&e\rightarrow_t\!n_i$. The type subscript $t$ can be omitted if there is no ambiguity caused. These functions model memory layout of record types. Intuitively speaking, if $e$ is a non-nil reference to a record-typed data object, $\&e\rightarrow_t\!n_i$ is the reference to the field $n_i$.
\end{enumerate}

The above function (and constant) symbols can be used in both programs and specifications. For conciseness, we use the following
abbreviations.
\begin{enumerate}
\item Let $v$ be a program variable declared with type \textbf{integer}, \textbf{boolean} or a pointer type, $v$ is an abbreviation for $\ast(\&v)$.
\item For a program variable $v$ declared with an array type $\textbf{ARR}(t,c)$, and $t$ is \textbf{integer}, \textbf{boolean}, or a pointer type, we use $v[e]$ as an abbreviation for $\ast(\&(\&v)[e])$.
\item If $e$ is of type $\textbf{P}(t)$, $t$ is a record type of which $n$ is a field name, and the field type is \textbf{integer}, \textbf{boolean} or a pointer type, we can use $e\rightarrow n$ as an abbreviation for  $\ast(\& e\rightarrow n)$.

\item Let $v$ be a program variable declared with a record type of which $n$ is a field name, the field type is \textbf{integer}, \textbf{boolean} or a pointer type, we can use $v.n$ as an abbreviation for $\ast (\&(\&v)\rightarrow n)$.
\end{enumerate}

\subsection{The proof rules about memory access and layout}\label{NEW-SYMB}
In this subsection, we present some proof rules to specify memory unit access and memory layout of composite types.
We define an auxiliary function $\texttt{Block}:\textbf{Ptr}\rightarrow \textbf{SetOf}(\textbf{Ptr})$ to denote the set of memory units in a memory block. The definition of $\texttt{Block}$ is as follows.

\noindent
$\texttt{Block}(r)= \emptyset$ if $r=\textbf{nil}$. Otherwise $\texttt{Block}(r)=$
$$\left\{
\begin{array}{rcl}
\{r\}  &\ \   & \mbox{if $\ast r$ is of type $\textbf{integer}$, $\textbf{boolean}$ or $\textbf{Ptr}$}\\
\bigcup_{\mbox{\tiny $n$: field name of $t$}}\texttt{Block}(\&r\rightarrow n)&\ \
&\mbox{if $r:\textbf{P}(t)$ and $t$ is a record type}\\
\bigcup_{i=0}^{c-1} \texttt{Block}(\&r[i]) &\ \ &\mbox{if $r:\textbf{P}(\textbf{ARR}(t',c))$ for some $t'$}
\end{array}
\right. $$
Intuitively speaking, $\texttt{Block}(r)$ is the set of memory units in the memory block referred by $r$.

The rule MEM-ACC says that if $r$ denotes a non-nil pointer referring to a memory unit storing basic type values or  pointer values, $\ast r$ denotes a basic type value or a pointer value respectively.
$$\framebox{\ \ \ MEM-ACC\ \ \ }\frac{\ \ \ \ \ r:\textbf{P}(t)\ \ \  r\neq \textbf{nil}\ }{\ast r : t}
\mbox{\small\ \ $t$ is $\textbf{integer}$, $\textbf{boolean}$, or $\textbf{P}(t')$ for some $t'$}$$\\

The rule MEM-BLK specifies how memory blocks are allocated. Given two arbitrary different memory blocks, they are either disjoint with each other, or one is contained by the other.
$$\framebox{MEM-BLK}\frac{p:\textbf{Ptr}\ \ \ q:\textbf{Ptr}\ \ \ p\neq q}
{\begin{array}{c}\texttt{Block}(p)\cap
\texttt{Block}(q)=\emptyset\lor\\
\texttt{Block}(p)\subset \texttt{Block}(q) \lor
\texttt{Block}(q)\subset \texttt{Block}(p)\end{array}}$$\\

The following two rules specify how the memory blocks are allocated for declared program variables. The rule PVAR-1 says that
for each program variable, a memory block with corresponding type is allocated. Furthermore, this block is not a sub-block of any other memory block.
The rule PVAR-2 says that each program variable is allocated a separate memory block.

$$
\framebox{PVAR-1}\frac{\ \ \ \ }
{\begin{array}{c}\&v:\textbf{P}(t)\land\&v\neq \textbf{nil}\land\\
\forall x:\textbf{Ptr}\cdot \texttt{Block}(\&v)\not\subset \texttt{Block}(x)
\end{array}}
\mbox{\small\ $v$ is a program declared with type $t$.}
$$\\
$$\framebox{PVAR-2}\frac{ \ \ \  }{\ \ \ \&v_1\neq \&v_2\ \ \ \  }\mbox{\small\ \ $v_1,v_2$ are two different program variables}$$\\

The following two rules specify the memory layout for record-typed memory blocks.
The rule RECORD-1 says that a record-typed memory block is allocated as a whole, i.e. when a record-typed memory block is allocated, all the memory blocks for its fields are also allocated.
The rule RECORD-2 says that the memory blocks allocated for the fields are disjoint with each other.
$$\framebox{RECORD-1}\frac{\ \ \ r:\textbf{P}(\textbf{REC}(\dots\times(n,t)\times\dots))\ \ \ \ r\neq \textbf{nil}}
{\ \ \ \ (\&r\rightarrow n:\textbf{P}(t))\land(\&r\rightarrow n\neq\textbf{nil})}$$\\

$$\framebox{RECORD-2}\frac{\ \ \ r:\textbf{P}(\textbf{REC}(\dots\times(n_1,t_1)\times\dots\times(n_2,t_2)\times\dots))\ \ \ \ r\neq \textbf{nil}}{\ \ \ \ \ \ \ \texttt{Block}(\&r\rightarrow n_1)\cap \texttt{Block}(\&r\rightarrow
        n_2)=\emptyset\ \ \ }$$\\

The following two rules specify the memory layout for array-typed memory blocks.
The rule ARR-1 says that an array-typed memory block is allocated as a whole, i.e. when an array-typed memory block is allocated, all of the memory blocks for its elements are allocated. The rule ARR-2 says that the memory blocks allocated for different elements are disjoint with each other.
$$\framebox{ARR-1}\frac{\ \ \ r:\textbf{P}(\textbf{ARR}(t,c))\ \ \ r\neq \textbf{nil}\ \ \  0\le i< c\ \ \ }{(\&r[i]:\textbf{P}(t))\land(\&r[i]\neq\textbf{nil})}$$\\

$$\framebox{{ARR}-2}\frac{\ \ \ r:\textbf{P}(\textbf{ARR}(t,c))\ \ \ r\neq \textbf{nil}\ \ \  0\le i< c\ \ \ 0\le j<c\ \ \ i\neq j}{\texttt{Block}(\&r[i])\cap \texttt{Block}(\&r[j])=\emptyset}$$\\

\subsection{The interpretation of P-types and new function symbols}
Please be noticed that the types of the constant symbols ($\&v, \textbf{nil}_t$ ) introduced in this section are \textbf{integer}, \textbf{boolean}, or pointer types. The argument types and result types of the function symbols introduced in this section are also \textbf{integer}, \textbf{boolean}, and pointer values.
So the terms in our logic do not denote array or record P-type values. Thus structures for our logic does not have to interpret record and array types.

For each P-type $t$, $(\textbf{P}(t))^A$ is a countable infinite set in the universal domain $\mathcal{U}^A$ satisfying that $(\textbf{nil}_{\textbf{P}(t)})^A\in (\textbf{P}(t))^A$. Furthermore, it is required that for different P-types $t_1$ and $t_2$, $( \textbf{P}(t_1))^A$ and $( \textbf{P}(t_2))^A$ are disjoint. $\textbf{Ptr}^A$ is the union of all such sets.

The function symbols $\&\!\rightarrow\!n$ and $\&[\,]$ model the memory layout of records and arrays respectively. As we do not go into details about memory layout of composite types, we just requires that all the proof rules in the previous subsection are satisfied by the interpretation of these function symbols.

The function symbol $\ast$ models program states. Its interpretation must satisfy that
$\ast^A(x) \in t^A$ if $x\in (\textbf{P}(t))^A$ and $x\neq (\textbf{nil}_{\textbf{P}(t)})^A$, where $t$ is \textbf{integer}, \textbf{boolean}, or $\textbf{P}(t')$ for some $t'$; $\ast^A(x) = \bot$ otherwise.

\section{Memory scope functions}\label{SUBSEC-SCOPE}
In LPF, a formula $f(x_1,\dots, x_n)\triangleq e$ defines a function denoted by $f$.
People can define data-retrieve functions using such formulae. In definitions for DRFs, we require that for any conditional sub-term
$e_0?e_1:e_2$ of $e$, none of the function symbols occurred in $e_0$ is defined (directly or indirectly) based on $f$. So for each DRF definition
$f(x_1,\dots, x_n)\triangleq e$, $e$ is continuous in $f$, thus we can use the proof rule \texttt{Func-Ind} in \cite{LPF} to prove properties about DRFs.

Given an LPF term $e$, the memory scope form  of $e$, denoted as $\mathfrak{M}(e)$,  is defined as follow.
\begin{enumerate}
\item If $e$ is a variable, $\mathfrak{M}(e)$ is $\emptyset$.
\item If $e$ is of the form $f(e_1,\dots,e_n)$, $\mathfrak{M}(e)$ is $\mathfrak{M}(e_1)\cup\dots\cup \mathfrak{M}(e_n)\cup \mathfrak{M}(f)(e_1,\dots,e_n)$, where $\mathfrak{M}(f)$ represents the MSF symbol of $f$, which is defined as follow
    \begin{itemize}
    \item If $f$ is a function symbol associated with basic types or abstract types (for example, $+,-,\times, /, >, <, \in, \subseteq \dots$), $\mathfrak{M}(f)$ is defined as the constant $\emptyset$.
    \item If $f$ is $\&\rightarrow n$, $\&[\,]$, $\&v$ for some program variable, $\textbf{nil}_t$ for some type $t$, $\mathfrak{M}(f)$ is defined as the constant $\emptyset$.
    \item If $f$ is the memory access function $\ast$ introduced in sub-section~\ref{SUB-SEC-NEWFUNC}, $\mathfrak{M}(\ast)$ is defined as $\mathfrak{M}(\ast)(x) \triangleq x$.
    \item For any other function symbols, $\mathfrak{M}(f)$ represents a new function symbol denoting the memory scope function of $f$.
    \end{itemize}
\item If $e$ is of the form $e_0 ? e_1: e_2$,  $\mathfrak{M}(e)$ is $\mathfrak{M}(e_0)\cup(e_0? \mathfrak{M}(e_1) :\mathfrak{M}(e_2))$.
\end{enumerate}

Given a DRF $f$ defined as $f(x_1,\dots, x_n)\triangleq e$, the memory scope function $\mathfrak{M}(f)$ of $f$ is defined as
        $$\mathfrak{M}(f)(x_1,\dots,x_n)=\mathfrak{M}(e)$$
Formally, it is expressed using the following proof rule.

$$\framebox{SCOPE-FUNC}\frac{\ \ \ \ \ \ f(x_1, \dots, x_n)\triangleq e}{\ \ \ \ \ \ \ \mathfrak{M}(f)(x_1, \dots, x_n)\triangleq\mathfrak{M}(e)}\ \ $$
Please be noticed that for any sub-term $e_0?e_1:e_2$ of $\mathfrak{M}(e)$, no function symbol is recursively defined based on $\mathfrak{M}(f)$.
\begin{definition}
 We say a structure $A$ with signature $\Sigma$  \emph{conforms} to a set of function definitions $\mathbb{P}$ iff for each definition $f(x_1,\dots,x_n)\triangleq e$ in $\mathbb{P}$, $\seml f(x_1,\dots,x_n)\triangleq e\semr^A_\alpha$ is $T$, here $\alpha$ is the assignment of $A$.
\end{definition}
The structure $A$ conforms to $\mathbb{P}$ means that $A$ interprets the defined function symbols according to their definition in $\mathbb{P}$.

In our logic, the function symbol $\ast$ is used to model program states. The DRFs used to retrieve abstract values are defined on $\ast$.
One of the basic ideas of our logic is that the abstract values retrieved by these functions keep unchanged if no memory unit in their memory scopes is over-written during a program execution. We have the following lemma and theorem about MSFs and memory scope forms of terms.

\begin{lemma}\label{THEOREM-TERM-EQUIV}
Let $\mathbb{P}$ be a set of recursive function definitions. Let $A$ and $A'$ be two structures. They both conform  to $\mathbb{P}$ and are identical except that they may have different interpretations for $\ast$ and for the function symbols defined in $\mathbb{P}$.
Let $e$ be a term satisfying that all function symbols in $e$ are either defined in $\mathbb{P}$, or associated with basic types, abstract types or P-types.
We have that $\seml e\semr^A_\alpha = \seml e\semr^{A'}_\alpha$ and $\seml \mathfrak{M}(e)\semr ^{A}_\alpha=\seml \mathfrak{M}(e)\semr ^{A'}_\alpha$ if
$\seml \mathfrak{M}(e)\semr^{A}_\alpha\neq \bot$ and $\ast^A(x)=\ast^{A'}(x)$ for all $x\in\seml \mathfrak{M}(e)\semr^{A}_\alpha$.

\end{lemma}
\begin{proof}

By induction, we first prove that the conclusion holds when $e$ contains no function symbol defined in $\mathbb{P}$.\\
\textbf{BASE}: The conclusion holds if $e$ is a variable or a constant symbol.\\
\textbf{INDUCTION}: Assuming the conclusion holds for all terms shorter than $e$. We prove that $\seml e\semr^A_\alpha = \seml e\semr^{A'}_\alpha$ and $\seml \mathfrak{M}(e)\semr ^{A}_\alpha=\seml \mathfrak{M}(e)\semr ^{A'}_\alpha$ if
$\seml \mathfrak{M}(e)\semr^{A}_\alpha\neq \bot$ and $\ast^A(x)=\ast^{A'}(x)$ for all $x\in\seml \mathfrak{M}(e)\semr^{A}_\alpha$.
\begin{itemize}
\item If $e$ is of the form $f(e_1,\dots, e_n)$, here $f$ is a function symbol other than $\ast$, and $f$ is not defined in $\mathbb{P}$. $\mathfrak{M}(e)$ is $\mathfrak{M}(e_1)\cup\dots\cup\mathfrak{M}(e_n)$ because $\mathfrak{M}(f)$ is $\emptyset$. So $\seml \mathfrak{M}(e_i)\semr^A_\alpha\neq\bot$ and $\ast^{A}(x)=\ast^{A'}(x)$ for all $x\in\seml\mathfrak{M}(e_i)\semr^{A}_\alpha$ for $i=1,\dots,n$.
    According to the inductive assumption, $\seml e_i\semr ^A_\alpha=\seml e_i\semr ^{A'}_\alpha$ and $\seml \mathfrak{M}(e_i)\semr^{A}_\alpha = \seml \mathfrak{M}(e_i)\semr^{A'}_\alpha$. It follows that $\seml e\semr ^A_\alpha=\seml e\semr ^{A'}_\alpha$ and $\seml \mathfrak{M}(e)\semr ^A_\alpha=\seml\mathfrak{M}(e)\semr ^{A'}_\alpha$ because $f^A=f^{A'}$.
\item If $e$ is of the form $\ast e_1$. $\mathfrak{M}(e)$ is defined as $\{e_1\}\cup\mathfrak{M}(e_1)$. So $\seml \mathfrak{M}(e_1)\semr^A_\alpha\neq\bot$ and $\ast^A(x)=\ast^{A'}(x)$ for all $x\in\seml \mathfrak{M}(e_1)\semr^A_\alpha$. From the inductive assumption, we have that $\seml e_1\semr ^A_\alpha =\seml e_1\semr ^{A'}_\alpha$ and $\seml\mathfrak{M}(e_1)\semr ^A_\alpha=\seml\mathfrak{M}(e_1)\semr ^{A'}_\alpha$. Because $\seml e_1\semr ^A_\alpha\in \seml\mathfrak{M}(e)\semr ^A_\alpha$ if $\seml\mathfrak{M}(e)\semr ^A_\alpha$ is not $\bot$, we have $\seml e\semr ^A_\alpha = \ast^A(\seml e_1\semr^A_\alpha) = \ast^{A'}(\seml e_1\semr^{A'}_\alpha) = \seml e\semr ^{A'}_\alpha$ and $\seml\mathfrak{M}(e)\semr ^A_\alpha = \seml\mathfrak{M}(e)\semr ^{A'}_\alpha$.
\item If $e$ is of the form $e_0?e_1:e_2$. $\mathfrak{M}(e)$ is $\mathfrak{M}(e_0)\cup(e_0?\mathfrak{M}(e_1):\mathfrak{M}(e_2))$. From the inductive assumption, we have $\seml e_0\semr ^A_\alpha=\seml e_0\semr ^{A'}_\alpha$ and  $\seml \mathfrak{M}(e_0)\semr ^A_\alpha = \seml \mathfrak{M}(e_0)\semr ^{A'}_\alpha$. So  $\seml e_0\semr ^{A'}_\alpha = T$ iff $\seml e_0\semr ^A_\alpha=T$. When $\seml e_0\semr ^{A'}_\alpha = \seml e_0\semr ^A_\alpha=T$, we have
    $\seml \mathfrak{M}(e)\semr ^A_\alpha = \seml\mathfrak{M}(e_0)\semr ^A_\alpha\cup\seml\mathfrak{M}(e_1)\semr ^A_\alpha$, $\seml\mathfrak{M}(e)\semr ^{A'}_\alpha = \seml\mathfrak{M}(e_0)\semr ^{A'}_\alpha\cup\seml\mathfrak{M}(e_1)\semr ^{A'}_\alpha$, $\seml e\semr ^A_\alpha = \seml e_1\semr ^A_\alpha$ and $\seml e\semr ^{A'}_\alpha = \seml e_1\semr ^{A'}_\alpha$. From the inductive assumption, we have that $\seml e\semr ^A_\alpha=\seml e\semr ^{A'}_\alpha$ and $\seml\mathfrak{M}(e)\semr ^{A}_\alpha=\seml\mathfrak{M}(e)\semr ^{A'}_\alpha$. We can also prove that $\seml e\semr ^A_\alpha=\seml e\semr ^{A'}_\alpha$ and $\seml\mathfrak{M}(e)\semr ^{A}_\alpha=\seml\mathfrak{M}(e)\semr ^{A'}_\alpha$ when $\seml e_0\semr ^A_\alpha$ is $F$ or $N$.
\end{itemize}

Second, we prove that the conclusion holds if no function symbol defined in $\mathbb{P}$ is (directly or indirectly) recursively defined on itself.
We give a rank to each term and each function symbol. The rank of a term $e$ is the highest rank of the function symbols occur in $e$.
The ranks of function symbols associated with basic types and abstract types are $0$. The function symbols $\ast$, $\&\rightarrow n$, $\&[]$ also have rank $0$.
 The rank of a function symbol $f$ defined as $f(x_1,\dots, x_n)\triangleq e_r$ in $\mathbb{P}$ is the rank of $e_r$ plus $1$.  As no function symbol is recursively defined, each function symbol and term has a rank. Now, the conclusion is proved by an induction on the ranks and the lengthes of terms.\\
\textbf{BASE}: According to the conclusion of the first step, this conclusion holds for $0$-rank terms with any length.\\
\textbf{INDUCTION}: Let $e$ be a $k$-rank term. If the conclusion holds for all terms either with a rank less than $k$, and all $k$-rank terms shorter than $e$.
\begin{itemize}
\item If $e$ is of the form $f(e_1,\dots,e_n)$ and $f$ is a function symbol with a rank non-greater than $k$, and defined as $f(x_1,\dots,x_n)\triangleq e_r$. Then the rank of $e_r$ is less than or equal to $k-1$. As all the function-definition formulae in $\mathbb{P}$ are interpreted to $T$, according to the semantic model of function definitions of LPF,  both $\seml f(e_1,\dots,e_n)\semr ^A_{\alpha}$ and $\seml f(e_1,\dots,e_n)\semr ^{A'}_{\alpha}$ are $\bot$ if some of $\seml e_i\semr ^A_\alpha$ is $\bot$. Otherwise,  $\seml f(e_1,\dots,e_n)\semr ^A_{\alpha}$ and  $\seml f(e_1,\dots,e_n)\semr ^{A'}_{\alpha}$ are $\seml e_r\semr ^A_{\alpha'}$ and $\seml e_r\semr ^{A'}_{\alpha'}$ respectively, where $\alpha'=\alpha(x_1\rightarrow \seml e_1\semr ^A_\alpha)\dots(x_n\rightarrow \seml e_n\semr ^A_\alpha)$, i.e. $\alpha'$ is same as $\alpha$ except that $\alpha'$ maps $x_i$ to $\seml e_i\semr^A_\alpha$; $\seml \mathfrak{M}(f)(e_1,\dots,e_n)\semr ^A_\alpha$ and $\seml \mathfrak{M}(f)(e_1,\dots,e_n)\semr ^{A'}_\alpha$ are $\seml \mathfrak{M}(e_r)\semr ^A_{\alpha'}$ and $\seml \mathfrak{M}(e_r)\semr ^{A'}_{\alpha'}$ respectively.
     Because $\seml \mathfrak{M}(f)(e_1,\dots,e_n)\semr ^A_\alpha\subseteq \seml e\semr^A_\alpha$, we have $\ast^A(x)=\ast^A(x)$ for all $x\in\seml \mathfrak{M}(f)(e_1,\dots,e_n)\semr ^A_\alpha=\seml \mathfrak{M}(e_r)\semr ^A_{\alpha'}$. As the rank of $e_r$ is less than or equal to $k-1$, from the inductive assumption, we have $\seml e_r\semr ^A_{\alpha'}=\seml e_r\semr ^{A'}_{\alpha'}$ and $\seml \mathfrak{M}(e_r)\semr ^A_{\alpha'}=\seml \mathfrak{M}(e_r)\semr ^{A'}_{\alpha'}$, i.e. $\seml e\semr ^A_{\alpha}=\seml e\semr ^{A'}_{\alpha}$ and  $\seml \mathfrak{M}(f)(e_1,\dots,e_n)\semr ^A_\alpha=\seml \mathfrak{M}(f)(e_1,\dots,e_n)\semr ^{A'}_\alpha$. So $\seml \mathfrak{M}(e)\semr ^A_{\alpha}=\seml \mathfrak{M}(e)\semr ^{A'}_{\alpha}$ because $\seml\mathfrak{M}( e_i)\semr^A_\alpha=\seml\mathfrak{M}( e_i)\semr^{A'}_\alpha$ and  $\seml \mathfrak{M}(f)(e_1,\dots,e_n)\semr ^A_\alpha=\seml \mathfrak{M}(f)(e_1,\dots,e_n)\semr ^{A'}_\alpha$.
\item If $e$ is a conditional form with a rank $k$, the proof is similar to those of the first step.
\end{itemize}

Now we are about to prove the general case. For each function symbol $f$ recursively defined in $\mathbb{P}$, we introduce infinite number of function symbols $f_0, f_1, \dots$. For the definition of $f$, i.e. $f(x_1,\dots, x_n)\triangleq e$, we introduce a set of definitions $f_i(x_1,\dots, x_n) \triangleq e_i$ for $i=1,2,\dots$, where $e_i$ is derived by replacing each function symbol $g$ recursively defined in $\mathbb{P}$ by $g_{i-1}$ (Here $g$ can be $f$) in $e$. We also introduce a function definition $f_0(x_1,\dots,x_n)=\circledast$ for each $f_0$. Notice that $f_i$s are not recursively defined. Furthermore, $\mathfrak{M}(e_i)$ is same as the term derived by replacing $g$ and $\mathfrak{M}(g)$ respectively by $g_{i-1}$ and $\mathfrak{M}(g_{i-1})$  in $\mathfrak{M}(e)$. Because it is required that for each definition $f(x_1,\dots,x_n)=e$ in $\mathbb{P}$, $e$ is continuous in $f$, we have that $f_i$ is less defined than $f_{i+1}$, i.e. $f_{i+1}(x_1,\dots,x_n)=f_i(x_1,\dots,x_n)$ if $f_i(x_1,\dots, x_n)$ is defined for any $x_1,\dots,x_n$. Because $\mathfrak{M}(e)$ is continuous in $\mathfrak{M}(f)$, we have that $\mathfrak{M}(f_i)^A$ is less defined than $\mathfrak{M}(f_{i+1})$. So $f^A$ is the least upper-bound of the function sequence $f_0^A, f_1^A,\dots$, and $\mathfrak{M}(f)^A$ is the least upper-bound of the function sequence $\mathfrak{M}(f_0)^A, \mathfrak{M}(f_1)^A,\dots$.
Let $e$ be a term containing recursively defined function symbols. If $\seml e\semr ^A_\alpha$ is not $\bot$, there must be a large-enough integer $i$ such that $\seml e_i\semr ^A_\alpha = \seml e\semr ^A_\alpha$ and $\seml \mathfrak{M}(e_i)\semr ^A_\alpha = \seml e\semr ^A_\alpha$, where $e_i$ is derived by replacing each recursively defined function symbol $g$ by $g_{i-1}$. As $e_i$ contains no recursively defined symbols, according to the second conclusion, we have that this lemma holds in general.
\hfill\textbf{QED}\\
\hfill$\square$
\end{proof}

The following theorem~\ref{FORMULA-EQUIV} gives a sufficient condition under which an LPF formula $p$ keeps unchanged before/after some memory units are modified. A term occurs in $p$ is called a top-level one if it is not a sub-term of another term occurs in $p$.

\begin{theorem}\label{FORMULA-EQUIV}
Let $\mathbb{P}$ be a set of recursive function definitions. Let $A$ and $A'$ be two structures. They both conform  to $\mathbb{P}$ and are identical except that they may have different interpretations for $\ast$ and for the function symbols defined in $\mathbb{P}$.
Let $p$ be an LPF formula satisfying that
\begin{itemize}
\item all function symbols in $p$ are either defined in $\mathbb{P}$, or associated with basic types, abstract types, or P-types, and
\item $p$ has no sub-formula of the form $f(x_1,\dots,x_n)\triangleq e_r$.
\end{itemize}
We have that $\seml p\semr ^A_\alpha=\seml p\semr ^{A'}_\alpha$ if $\seml \mathfrak{M}(e)\semr ^{A}_{\alpha'}\neq\bot$ and $\ast^A(x)=\ast^{A'}(x)$ for all $x\in\seml \mathfrak{M}(e)\semr ^{A}_{\alpha'}$ for each top-level term $e$ of $p$, and arbitrary assignment $\alpha'$.
\end{theorem}
\begin{proof}
This theorem can be proved by an induction on the structure of $p$.\\
\textbf{BASE:}
\begin{itemize}
\item If $p$ is of the form $f(e_1,\dots,e_n)$, and $f$ is a \textbf{boolean}-typed function symbol (or a predicate symbol). $p$ itself is the only top-level term of $p$. From Lemma~\ref{THEOREM-TERM-EQUIV}, $\seml p\semr ^A_\alpha=\seml p\semr ^{A'}_\alpha$.
\item If $p$ is of the form $e_1 = e_2$. From Lemma~\ref{THEOREM-TERM-EQUIV}, $\seml e_i\semr ^A_\alpha=\seml e_i\semr ^{A'}_\alpha$ for $i=1,2$. So $\seml p\semr ^A_\alpha=\seml p\semr ^{A'}_\alpha$.
\item If $p$ is of the form $e:t$. From Lemma~\ref{THEOREM-TERM-EQUIV}, $\seml e\semr ^A_\alpha=\seml e\semr ^{A'}_\alpha$ and $t^A=t^{A'}$. So $\seml p\semr ^A_\alpha=\seml p\semr ^{A'}_\alpha$.
\end{itemize}
\textbf{INDUCTION:}
\begin{itemize}
\item If $p$ is of the form $\forall x:t\cdot p'$. A top-level term of $p$ is also a top-level term of $p'$. From the inductive assumption, for an assignment $\alpha(x\rightarrow v)$ for an arbitrary $t$-typed value $v$, $\seml p'\semr ^A_{\alpha(x\rightarrow v)}=\seml p'\semr ^{A'}_{\alpha(x\rightarrow v)}$. According to the interpretation rule for $\forall x:t\cdot p'$, we conclude that
    $\seml p\semr ^A_{\alpha}=\seml p\semr ^{A'}_{\alpha}$.
\item The conclusion can also be proved when $p$ is of the form $\Delta p'$, $\neg p'$, and $p_1\land p_2$.
\end{itemize}
\hfill\textbf{QED}\\
\hfill$\square$
\end{proof}

\section{Syntax of programs}\label{SEC-PROG-SYNTAX}
The small program language used in this paper is strong typed. Each expression in the programs has a static P-type. An expression $e$ has a static P-type $t$ means that at the runtime, either $e$  denotes a value of type $t$ or $e$ is non-denoting. The argument types and result types of function symbols appeared in programs are definitely specified. The static types of expressions can be decided statically and automatically.  It also can be statically checked (by a compiler, for example) that each function symbol is applied to arguments with suitable static types. In this paper, it is supposed that all programs under verification have passed such static type check.

\subsection{The syntax of program expressions}
A program expression is an LPF term with following restrictions.
\begin{enumerate}
\item A program expression contains no free variable. Be noticed that a program variable $v$ occurs in a term is in fact an abbreviation for $\ast(\&v)$.
\item Only the following function (predicate) symbols can occur in program expressions.
\begin{enumerate}
\item Constant symbols for basic types (\textbf{integer}, \textbf{boolean}), $\textbf{nil}_t$ for type $t$, $\&v$ for a program variable $v$;
\item Function symbols associated with \textbf{integer} and \textbf{boolean}, like $+,-, *, \div, <, \le, =, \dots$;
\item Memory access/layout function symbols $\ast$, $\&\rightarrow n$, $\&[\,]$;
\item Boolean functions $\textbf{not}$, $\textbf{cand}$, $\textbf{cor}$ which are defined using conditional forms as follows.
 \begin{enumerate}
 \item $\textbf{not}\ x \triangleq x \mbox{?} \textbf{false} : \textbf{true}$
 \item  $x \ \textbf{cand} \ y\triangleq \neg x \mbox{?} \textbf{false} : y$
 \item $x \ \textbf{cor} \ y\triangleq x \mbox{?} \textbf{true} : y$.
 \end{enumerate}
\end{enumerate}
\end{enumerate}
We define these boolean operators because the semantic of logical connectives $\land$ and $\lor$ of LPF is different from that of the logical operators commonly used in program languages.

\subsection{The syntax of program statements}
The syntax of program statements is as follows.
$$\begin{array}{rcl}
st & ::=\ \ &\texttt{skip}\ \ |\ \ \ast e_1:=e_2\ \ |\ \ \ast e:=\texttt{alloc}(t)\ \ \\
   &        &|\ \ st;\ st\ \ |\ \ \textbf{if}\ (e) \ st\ \textbf{else}\ st\\
   &        &|\ \ \textbf{while}\ (e)\ st\\
\end{array}$$
This programming language has two kinds of primitive statements:
assignment statements and memory-allocation statements.
\begin{itemize}
\item An assignment statement $\ast e_1:=e_2$ first evaluates $e_1$ and $e_2$, then
assigns the value of $e_2$ to the memory unit referred by the value of $e_1$. The values stored in other memory units keep unchanged.
It is required that $\ast e_1$ and $e_2$ has same static type, which is limited to be \textbf{integer}, \textbf{boolean}, or a pointer type.
\item A memory-allocation statement
$\ast e:= \texttt{alloc}(t)$ allocates a memory block of type $t$, and
assigns the reference to this memory block to the memory unit referred by
the value of $e$. Furthermore, in the new memory block, all the memory units storing pointer values are initialized to \textbf{nil}. It is required that the static type of $\ast e$ is $\textbf{P}(t)$.
\end{itemize}

The semantics of the composite statements $st;st$, $\textbf{if}\ (e) \ st\ \textbf{else}\ st$, and $\textbf{while}\ (e)\ st$ are same as those commonly used in real program languages. It is required that in $\textbf{if}\ (e) \ st\ \textbf{else}\ st$ and $\textbf{while}\ (e)\ st$, the static type of $e$ must be \textbf{boolean}.

\begin{example}
The program depicted in Figure~\ref{EXAMPLE} is a running example
used in this paper. The type of the program variables $\textsf{k}$ and
$\textsf{d}$ is \textbf{integer}. The type of program variables $\textsf{root}$ and $\textsf{p}$ is
\textbf{P}($T$), where $T$ is $\textbf{REC}((l, \textbf{P}(T))\times(r,\textbf{P}(T))\times (K,\textbf{integer})\times(D,\textbf{integer}))$. This program
first searches a binary search tree for a node of which the field $K$ equals $\textsf{k}$. Then it sets the
filed $D$ of this node to $\textsf{d}$. Please be noticed that $\textsf{p}$, $\textsf{root}$, $\textsf{k}$, $\textsf{d}$, $\textsf{p}\rightarrow K$, $\textsf{p}\rightarrow D$, $\textsf{p}\rightarrow l$, $\textsf{p}\rightarrow r$ are respectively abbreviations for
$\ast(\&\textsf{p})$,  $\ast(\&\textsf{root})$, $\ast(\&\textsf{k})$, $\ast(\&\textsf{d})$, $\ast(\&\textsf{p}\rightarrow K)$, $\ast(\&\textsf{p}\rightarrow D)$, $\ast(\&\textsf{p}\rightarrow l)$, $\ast(\&\textsf{p}\rightarrow r)$.
\end{example}

\begin{figure}
\begin{center}
\parbox{300pt}{
\begin{tabbing}
\textsf{p}:=\textsf{root};\\
\textbf{while} \= ($\textsf{p}\rightarrow K \neq \textsf{k}$)\\
\{\\
\>\textbf{if} ($\textsf{k} < \textsf{p}\rightarrow K$ ) $\textsf{p} := \textsf{p}\rightarrow l$ \textbf{else} $\textsf{p} := \textsf{p}\rightarrow r$;\\
 \}\\
$\textsf{p}\rightarrow D := \textsf{d}$;
\end{tabbing}}
\end{center}
\caption{The program used as a running example}\label{EXAMPLE}
\end{figure}

\section{Syntax of specifications}\label{SEC-SYN-SPEC}
A program specification is of the form $\mathbb{P}\vdash q \{c\} r$, where  $c$ is a program, $\mathbb{P}$ is a set of LPF formulae, $q$ and $r$ are LPF formulae satisfying the following conditions.
 \begin{itemize}
 \item They contain only function symbols defined in $\mathbb{P}$, the function symbols which can occur in program expressions, and the function symbols associated with abstract types.
 \item $q$ and $r$ contains no sub-formula of the form $f(x_1,\dots,x_n)\triangleq e$.
 \end{itemize}
The formula set $\mathbb{P}$ is called the premise of this specification. $\mathbb{P}$ usually contains a set of function definitions. The formulae $q$ and $r$ are respectively called  the pre-condition and post-condition. Intuitively speaking, such a specification means that if all the formulae in $\mathbb{P}$ hold for arbitrary program states, and the program $c$ starts its execution on a state satisfying $q$, then the state satisfies $r$ when the program $c$ stops.
\begin{example}
Let $\mathbb{P}$ be the set of formulae depicted in Figure~\ref{DATA-STRUCTURE-INTERPRETATION-FUNCTIONS}. These formulae define a set of data retrieve functions. The boolean function $\texttt{InHeap}$ is defined in sub-section~\ref{SUB-SEC-MEM-ALLOC}. $\texttt{InHeap}(x)$ means that $x$ refers to a memory block disjoint with all memory blocks for program variables. Let
$$q = \textsf{isHBST}(\textsf{root})\land \textsf{Map}(\textsf{root})=M \land \textsf{k} \in \textsf{Dom}(\textsf{root})$$
$$r=\textsf{isHBST}(\textsf{root})\land \textsf{Map}(\textsf{root})=M\dag\{\textsf{k}\mapsto \textsf{d}\}$$
$Prog$ is the program depicted in Figure~\ref{EXAMPLE}. The specification  $\mathbb{P}\vdash q \{Prog\} r$ says that if the program state satisfies the following conditions when $\{Prog\}$ starts.
\begin{enumerate}
\item The value of $\textsf{root}$ points to the root node of a binary search tree stored in the heap;
\item The tree represents a finite map $M$ from \textbf{integer} to \textbf{integer};
\item The value stored in $\textsf{k}$ is in the domain of this map,
\end{enumerate}
When $Prog$ stops, $\textsf{root}$ still points to the root node of the binary search tree, and now the finite map represented by the binary search tree becomes $M\dag \{\textsf{k}\mapsto\textsf{d}\}$.
\end{example}

\begin{figure}
\begin{tabbing}
$\textsf{NodeSet}(x):\textbf{P}(T)\rightarrow \textbf{SetOf}(\textbf{Ptr})$\\
\ \ \ \ \ \ \ \ \=$\triangleq$ $(x=\textbf{nil}) ? $\=\ $\emptyset : (\{x\}\cup \textsf{NodeSet}(x\rightarrow l)\cup \textsf{NodeSet}(x\rightarrow r))$\\
\\
$\textsf{Map}(x):\textbf{P}(T)\rightarrow \textbf{Map integer to integer}$\\
        \>$\triangleq(x=\textbf{nil}) ? \emptyset :\{x\rightarrow K\mapsto x\rightarrow D\}\dag \textsf{Map} (x\rightarrow l)\dag \textsf{Map}(x\rightarrow r)$\\
\\
$\textsf{MapP}(x,y):\textbf{P}(T)\times\textbf{P}(T)\rightarrow \textbf{Map integer to integer}$\\
        \>$\triangleq(x=\textbf{nil}) ? \emptyset : \textsf{MapP}(x\rightarrow l) \dag \textsf{MapP} (x\rightarrow r)\dag$\\
           \>\ \ \ \ \ \ \ $((x = y)? \emptyset : \{x\rightarrow K\mapsto x\rightarrow D\})$\\
\\
$\textsf{Dom}(x):\textbf{P}(T)\rightarrow\textbf{SetOf}(\textbf{integer})$\\
        \>$\triangleq (x=\textbf{nil}) ? \emptyset : (\{x \rightarrow K\} \cup \textsf{Dom}(x\rightarrow l)\cup \textsf{Dom}(x\rightarrow r))$\\
\\
$\textsf{isHBST}(x):\textbf{P}(T)\rightarrow \textbf{boolean}$\\
        \>$\triangleq (x=\textbf{nil}) ? \textbf{true} : \texttt{InHeap}(x)\land \textsf{isHBST}(x\rightarrow l)\land \textsf{isHBST}(x\rightarrow r)\land$\\
         \>         \> $(\textsf{Dom}(x\rightarrow l)=\emptyset?\texttt{true}: \texttt{MAX}(\textsf{Dom}(x\rightarrow l))<x\rightarrow K) \land $\\
         \>         \>$(\textsf{Dom}(x\rightarrow r)=\emptyset?\texttt{true}:x\rightarrow K < \texttt{MIN}(\textsf{Dom}(x\rightarrow r)))$\\
\end{tabbing}
\caption{The definitions of a set of data retrieve functions}\label{DATA-STRUCTURE-INTERPRETATION-FUNCTIONS}
\end{figure}

\section{Proof rules of program statements}\label{SEC-AXIOM}
In this section, we present the proof rules for program statements.
There are three rules for primitive statements, one rule for consequences, and three rules for control flow statements.

\subsection{The proof rule for \textbf{skip} statement}
The skip statement changes nothing, so we have the following proof rule.

$$\framebox{SKIP-ST}\frac{\ \ \ \ }
{\ \ \ \ \ \ \ \ \ \emptyset \vdash q\{\texttt{skip}\}q\ \ \ \ \ \ \ \ \ \ }$$

\subsection{The proof rule for assignment statements}
Let $q$ be an LPF formula and $x$ be the only free variable in $q$. Let $t$ be the static type of $\ast e_1$ and $e_2$. The type $t$ must be \textbf{integer}, \textbf{boolean}, or $\textbf{P}(t')$ for some $t'$.
We have the following proof rule for assignment statements.

$$\framebox{ASSIGN-ST}\frac{\begin{array}{l}\mathbb{P}, q[e_2/x]\vdash e_1\neq \textbf{nil}\land e_1 \not\in \mathfrak{M}(e_1)\land e_2:t\\
   \mathbb{P}, q[e_2/x]\vdash e_1\not\in\mathfrak{M}(e)[e_2/x] \mbox{ for each top-level term $e$ of $q$}
   \end{array}
   }{\mathbb{P}\vdash q[e_2/x]\{\ast e_1:=e_2\}q[\ast e_1/x]}$$\\
Here, it is required that all bounded variables in $q$ are different from $x$. A term $e$ of $q$ is called a top-level one if it is not a sub-term of another term of $q$.
Furthermore, it is required that for each conditional term $e_0?e_1:e_2$ of $q$, $e_0$ is a boolean-typed term, so we can construct a memory form of each top-level term of $q$.

Now we briefly prove the soundness of this rule. We can use two structure $A$
and $A'$ to denote the program states before/after the assignment statement. $A$ and $A'$ are only different in the interpretations of the function symbol $\ast$ and the symbols defined in $\mathbb{P}$.
The semantic of an assignment $\ast e_1=e_2$ is as follow. It first
evaluates the value of $e_1$ and $e_2$, i.e. $\seml e_1\semr^A_\alpha$ and $\seml e_2\semr^A_\alpha$, then the content of
the memory unit referred by $\seml e_1\semr^A_\alpha$ is set to $\seml e_2\semr^A_\alpha$.
Formally, we say $\ast^{A'}(\seml e_1\semr ^A_\alpha) = \seml e_2\semr ^A_\alpha$, and $\ast^A(x)=\ast^{A'}(x)$ for all $x\neq \seml e_1\semr ^A_\alpha$.
According to Lemma~\ref{THEOREM-TERM-EQUIV}, the condition $e_1\not\in\mathfrak{M}(e_1)$
assures that $\seml e_1\semr ^A_\alpha = \seml e_1\semr ^{A'}_\alpha$, so $\seml e_2\semr ^A_\alpha = \ast^{A'}(\seml e_1\semr ^A_\alpha) =
\ast^{A'}(\seml e_1\semr ^{A'}_\alpha) = \seml \ast e_1\semr ^{A'}_\alpha$.
The condition $e_1\neq \textbf{nil}\land e_2:t$ assures that both $\seml \ast e_1\semr ^A_\alpha$ and $\seml e_2\semr ^A_\alpha$ are not $\bot$. Together with these conditions, the condition  $e_1\not\in \mathfrak{M}(e)[e_2/x]$  assures that for each top-level term $e$, $\seml e_1\semr^A_\alpha\not\in \seml\mathfrak{M}(e)[e_2/x]\semr^A_\alpha$, which equals to $\seml\mathfrak{M}(e)\semr^A_{\alpha(x\rightarrow \seml e_2\semr^A_\alpha)}$. From Lemma~\ref{THEOREM-TERM-EQUIV}, we have $\seml e\semr ^A_{\alpha(x\rightarrow \seml e_2\semr ^A_\alpha)}=\seml e\semr ^{A'}_{\alpha(x\rightarrow \seml e_2\semr ^A_\alpha)}=\seml e\semr ^{A'}_{\alpha(x\rightarrow \seml \ast e_1\semr ^{A'}_\alpha)}$. So we have
$\seml e[e_2/x]\semr ^A_\alpha=\seml e[\ast e_1/x]\semr ^{A'}_\alpha$. As $\alpha$ is arbitrary, according to Theorem~\ref{FORMULA-EQUIV},
$\seml q[e_2/x]\semr ^A_{\alpha}=\seml q[\ast e_1/x]\semr ^{A'}_{\alpha}$. So we conclude that if $q[e_2/x]$ holds before the assignment statement, $q[\ast e_1/x]$ holds after.

\subsection{The proof rule for memory allocation statements}\label{SUB-SEC-MEM-ALLOC}
The memory allocation statement $\ast e=\textsf{alloc}(t)$ first evaluates $e$, then allocates an unused memory block and assigns the reference to this block to the memory unit referred by $e$. All the memory units storing pointer values are initialized to $\textbf{nil}$. This block can not be referred by any pointers stored somewhere before this allocation. Furthermore, this block is disjoint with all of the memory blocks allocated for program variables. It is required that the static type of $\ast e$ must be $\textbf{P}(t)$.
Let  $p$ be an LPF formula containing no free variable, we have the following proof rule for memory allocation statements.

$$\framebox{ALLOC-ST}
\frac{
\begin{array}{l}
\mathbb{P}\land q\vdash e\neq \textbf{nil} \land e \not\in \mathfrak{M}(e)\\
\mathbb{P}\land q\vdash e\not\in \mathfrak{M}(e') \mbox{ for each top-level term $e'$ of $q$}
\end{array}
}
{
\mathbb{P}\vdash q\ \{*e=\textsf{alloc}(t)\}
\left(\begin{array}{l}
q \land \texttt{InHeap}(\ast e)\land \texttt{Unique}(e) \\
\land \texttt{PtrInit}(\ast e)
\land (\ast e\neq \textbf{nil})\end{array}\right)
}
$$
The predicts \texttt{Unique}, \texttt{InHeap}, and \texttt{PtrInit} are defined as follows.
$$\texttt{Unique}(x)\triangleq \forall y:\textbf{Ptr}\cdot ((y\neq x\land y\neq \textbf{nil}\land \ast y:\textbf{Ptr})\Rightarrow \texttt{Block}(\ast y)\cap \texttt{Block}(*x)=\emptyset)$$
$$\texttt{InHeap}(p)\triangleq  \bigwedge_{x\mbox{ is a program variable.}}(\texttt{Block}(\&v) \cap \texttt{Block}(p)=\emptyset)$$
$$\texttt{PtrInit}(p)\triangleq \forall x:\textbf{Ptr}\cdot ((x\in \texttt{Block}(p)\land x\neq \textbf{nil} \land \ast x:\textbf{Ptr}) \Rightarrow \ast x=\textbf{nil})$$
Intuitively speaking, $\texttt{Unique}(p)$ says that the memory block referred by the reference stored in $p$ can not be accessed by references stored elsewhere. $\texttt{InHeap}(p)$ says that the memory block referred by $p$ is disjoint with all the memory blocks for program variables.
$\texttt{PtrInit}(p)$ says that all memory units with pointer types in the memory block referred by $p$ store \textbf{nil} pointers.

Similarly to the soundness reasoning for the rule ASSIGN-ST, we can conclude that $q$ still holds after the allocation statement if it holds before. Because the allocated memory block is unused, it can not be accessed by any pointers stored somewhere before this memory allocation. This allocation statement assigns the reference to this block only to the memory unit referred by $e$. So $\texttt{Unique}(e)$ holds after this allocation statement.
The new allocated block is disjoint with any blocks for program variables. So $\texttt{InHeap}(\ast e)$ holds after this allocation statement.
The post condition $\texttt{PtrInit}(\ast e)$ holds because the new block is initialized as described above. So we conclude that this proof rule is sound.

\subsection{The consequence rule and the rules for control flow statements}
The following proof rules are essentially the same as those presented in \cite{HOARE-1}.
The consequence rule is slightly modified such that the premise of a verified assertion
can be strengthened. The rules for \texttt{if}-statement and \texttt{while}-statement are modified such that the pre-condition ensures that the condition expression $e$ is evaluated to either $T$ or $F$.

$$\framebox{CONSEQ} \frac{\ \ \ \mathbb{P}\vdash q\{s\}r \ \ \ \ \mathbb{Q}\vdash\mathbb{P}\ \ \ \ \ \mathbb{P}, q' \vdash q\ \ \ \ \ \mathbb{P},r\vdash r'\ \ \ }
{\mathbb{Q} \vdash q'\{s\}r'}$$\\

$$\framebox{SEQ-ST}\frac{\ \ \ \mathbb{P}\vdash q\{s_1\}r \ \ \ \ \ \mathbb{P}\vdash r\{s_2\}r'\ \ \ }
{\mathbb{P}\vdash q\{s_1;s_2\}r'}$$\\

$$\framebox{IF-ST}\frac{\ \ \ \mathbb{P},q\vdash e\lor\neg e\ \ \ \ \ \mathbb{P}\vdash (q \land e)\{ s_1\}r \ \ \ \ \mathbb{P}\vdash (q\land \neg e)\{s_2\} r\ \ \ }
{\mathbb{P}\vdash q\{\mbox{ \texttt{if} }(e)\ s_1\mbox{
\texttt{else} }s_2\ \}r}$$\\

$$\framebox{WHILE-ST}\frac{\ \ \ \ \mathbb{P},q\vdash e\lor\neg e\ \ \ \ \ \mathbb{P}\vdash (q \land e)\{s\}q}
{\ \ \ \mathbb{P}\vdash q\{\mbox{ \texttt{while} }(e)\ \ s\ \} q\land
\neg e\ \ \ }
$$

\section{Verifying the running example}\label{SEC-PROOF}
In this section, we verify the program depicted in Figure~\ref{EXAMPLE}.

\subsection{The DRFs, MSFs and their properties.}

\begin{example}
Figure~\ref{DATA-STRUCTURE-INTERPRETATION-FUNCTIONS} shows the
data-retrieve functions for specifying and verifying the program depicted in
Figure~\ref{EXAMPLE}. From the proof rule \textrm{SCOPE-FUNC} in Section~\ref{SUBSEC-SCOPE}, we can derive the definitions of all corresponding MSFs. The definitions of MSFs depicted in Figure~\ref{SCOPE-FUNCTIONS} are simplified but equivalent to those derived directly by the rule \textrm{SCOPE-FUNC}. For conciseness, we write $\mathfrak{M}(\textsf{NodeSet})$ as $\verb"NS"_m$,
$\mathfrak{M}(\textsf{Map})$ as $\verb"MP"_m$, $\mathfrak{M}(\textsf{MapP})$ as $\verb"MPP"_m$, $\mathfrak{M}(\textsf{Dom})$ as $\verb"DM"_m$, $\mathfrak{M}(\textsf{isHBST})$ as $\verb"HBST"_m$.  Some properties about these DRFs and MSFs are depicted in Figure~\ref{DRF-MSF-PROPERTIES}. These properties can be proved in the extended LPF.
\end{example}

\begin{figure}
\begin{tabbing}
$\verb"NS"_m(x)$\=$\triangleq$ $(x=\textbf{nil}) ? $\=\ $\emptyset : (\{\&x\rightarrow l, \&x\rightarrow r\}\cup \verb"NS"_m(x\rightarrow l)\cup \verb"NS"_m(x\rightarrow r))$\\
\\
$\verb"MP"_m(x)\triangleq(x=\textbf{nil}) ? \emptyset :$\\
           \>$\{\&x\rightarrow l, \&x\rightarrow r, \&x\rightarrow D, \&x\rightarrow K\}\cup \verb"MP"_m(x\rightarrow l)\cup \verb"MP"_m(x\rightarrow r)$\\
\\
$\verb"MPP"_m(x,y)\triangleq(x=\textbf{nil}) ? \emptyset :$\\
           \>$\{\&x\rightarrow l, \&x\rightarrow r\}\cup \verb"MPP"_m(x\rightarrow l) \cup \verb"MPP"_m(x\rightarrow r)\cup$\\
           \>$((x = y)? \emptyset : \{\&x\rightarrow K,\&x\rightarrow D\})$\\
\\
$\verb"DM"_m(x)\triangleq (x=\textbf{nil}) ? \emptyset : (\{\&x \rightarrow K, \&x\rightarrow l, \&x\rightarrow r\} \cup \verb"DM"_m(x\rightarrow l)\cup \verb"DM"_m(x\rightarrow r))$\\
\\
$\verb"HBST"_m(x)\triangleq (x=\textbf{nil}) ? \emptyset : \{\&x\rightarrow l, \&x\rightarrow r\}\cup \verb"HBST"_m(x\rightarrow l)\cup \verb"HBST"_m(x\rightarrow r)\cup$\\
         \>$\verb"DM"_m(x\rightarrow l)\cup (\textsf{Dom}(x\rightarrow l)=\emptyset?\emptyset: \{\&x\rightarrow l\}\cup \{\&x\rightarrow K\}\cup \verb"DM"_m(x\rightarrow l)) \cup $\\
         \>$\verb"DM"_m(x\rightarrow r)\cup (\textsf{Dom}(x\rightarrow r)=\emptyset?\emptyset: \{\&x\rightarrow r\}\cup \{\&x\rightarrow K\} \cup \verb"DM"_m(x\rightarrow r)))$\\
\end{tabbing}
\caption{The definitons of MSFs}\label{SCOPE-FUNCTIONS}
\end{figure}

\begin{figure}

\begin{equation}\mathbb{P}, \textsf{isHBST}(x)\vdash\&\textsf{p}\not\in\verb"HSBT"_m({x})\cup\verb"MP"_m({x})\cup\verb"DM"_m({x})\label{P-ISOLATE}\end{equation}

\begin{equation}\mathbb{P}, \textsf{isHBST}(x)\vdash\&\textsf{p}\rightarrow D \not\in\verb"HSBT"_m({x})\cup\verb"MPP"_m({x},\textsf{p})\cup\verb"DM"_m({x})\label{PD-ISOLATE}\end{equation}

\begin{equation}\mathbb{P}, \textsf{isHBST}(x), y\in \textsf{Dom}(x), y<x\rightarrow K \vdash y\in\textsf{Dom}(x\rightarrow l)\label{LEFT-TREV}\end{equation}

\begin{equation}\mathbb{P}, \textsf{isHBST}(x), y\in \textsf{Dom}(x), y>x\rightarrow K \vdash y\in\textsf{Dom}(x\rightarrow r)\end{equation}

\begin{equation}\label{MapPtEqMap}
\mathbb{P}, \textsf{isHBST}(x), \textsf{y} \in \textsf{NodeSet}(x)\vdash \textsf{Map}(x)=\textsf{MapP}(x,y)\dag
\{y\rightarrow K\mapsto y\rightarrow D\}\label{MAPMAPP}
\end{equation}

\begin{equation}\mathbb{P}, \textsf{NodeSet}(x):\textbf{SetOf}(\textbf{Ptr})\vdash x\in \textsf{NodeSet}(x)\label{NODE-SET-CONTAIN}\end{equation}

\caption{Some properties about DRFs and MSFs}\label{DRF-MSF-PROPERTIES}
\end{figure}

\subsection{Verifying the program}
In this section, we will prove that if $\textsf{root}$ points to a binary search tree, and we view this binary tree as a finite map, and $\textsf{k}$ is in the domain of this map, the program depicted in Figure~\ref{EXAMPLE} set the co-value of $\textsf{k}$ to $\textsf{d}$. In this section, we use $\mathbb{P}$ to denote the set of the function definitions in Figure~\ref{DATA-STRUCTURE-INTERPRETATION-FUNCTIONS}. The specification is as follow.
$$
\mathbb{P}\vdash \textsf{PRE-COND} \ \ \{\textsl{Prog}\}\ \ \textsf{isHBST}(\textsf{root}) \land \textsf{Map}(\textsf{root})=M\dag\{\textsf{k}\mapsto \textsf{d}\}
$$
Here, $\textsf{PRE-CON}$ is the abbreviation for $\textsf{isHBST}(\textsf{root}) \land \textsf{Map}(\textsf{root})= M \land \textsf{k}\in\textsf{Dom}(\textsf{root})$, $M$ is a constant with type \textbf{Map integer to integer}. The verification steps are given below.

From \textrm{ASSIGN-ST}, \ref{P-ISOLATE} and $\&\textsf{p}\not\in\{\&\textsf{root}, \&\textsf{k}\}$:
\begin{equation}
\mathbb{P}\vdash\left(\begin{array}{l}
(\textsf{PRE-COND}\land x \in \textsf{NodeSet}(\textsf{root}) \land \textsf{k}\in \textsf{Dom}(x))[\textsf{root}/x]\\
\ \ \ \ \ \{ \textsf{p} = \textsf{root};\}\ \ \ \\
(\textsf{PRE-COND}\land x \in \textsf{NodeSet}(\textsf{root}) \land \textsf{k}\in \textsf{Dom}(x))[\textsf{p}/x]
\end{array}\right )
\label{assign-1-PREV}
\end{equation}
From $\land$-I, \ref{assign-1-PREV}, \textrm{CONSEQ}, \ref{NODE-SET-CONTAIN}:
\begin{equation}
\mathbb{P}\vdash \textsf{PRE-COND}\ \{ \textsf{p} = \textsf{root};\}\ \textsf{PRE-COND} \land \textsf{p}\in \textsf{NodeSet}(\textsf{root})\land \textsf{k}\in \textsf{Dom}(\textsf{p})
\label{SPEC-ASSIGN-1}\end{equation}
From \textrm{ASSIGN-ST}, \ref{P-ISOLATE}, and $\&\textsf{p}\not\in\{\&\textsf{root}, \&\textsf{k}\}$:
\begin{equation}
\mathbb{P}\vdash\left(
\begin{array}{l}
(\textsf{PRE-COND} \land  x\in \textsf{NodeSet}(\textsf{root}) \land \textsf{k}\in \textsf{Dom}(x))[\textsf{p}\rightarrow l/x]\\
\ \ \ \ \ \{\textsf{p}:=\textsf{p}\rightarrow l;\}\\
(\textsf{PRE-COND} \land  x\in \textsf{NodeSet}(\textsf{root}) \land \textsf{k}\in \textsf{Dom}(x))[\textsf{p}/x]
\end{array}\right)\label{PEQPL}
\end{equation}
From \ref{LEFT-TREV}, \textbf{substitution}:
\begin{equation}
\begin{array}{l}
\mathbb{P}, \textsf{PRE-COND}, \textsf{p}\in\textsf{NodeSet}(\textsf{root}), \textsf{k}\in \textsf{Dom}(\textsf{p}), \textsf{k}<\textsf{p}\rightarrow K\vdash \\
\ \ \ \ \ \textsf{p}\rightarrow l\in \textsf{NodeSet}(\textsf{root})\land \textsf{k}\in \textsf{Dom}(\textsf{p}\rightarrow l)
\end{array}\label{LTRAV1}
\end{equation}
From \ref{PEQPL}, \ref{LTRAV1}, and \textrm{CONSEQ}:
\begin{equation}
\mathbb{P}\vdash\left(
\begin{array}{l}
\textsf{PRE-COND} \land\textsf{p}\in \textsf{NodeSet}(\textsf{root}) \land \textsf{k}\in \textsf{Dom}(\textsf{p})\land \textsf{k}< \textsf{p}\rightarrow K\\
\ \ \ \ \ \{\textsf{p}:=\textsf{p}\rightarrow l;\}\\
\textsf{PRE-COND} \land \textsf{p}\in \textsf{NodeSet}(\textsf{root}) \land \textsf{k}\in \textsf{Dom}(\textsf{p})
\end{array}\right)
\label{SPEC-ASSIGN-2}\end{equation}
Similarly, we can prove:
\begin{equation}
\mathbb{P}\vdash\left(
\begin{array}{l}
\textsf{PRE-COND} \land \textsf{p}\in \textsf{NodeSet}(\textsf{root})\land\textsf{k}\in \textsf{Dom}(\textsf{p})\land \textsf{k}>\textsf{p}\rightarrow K\\
\ \ \ \ \ \{\textsf{p}:=\textsf{p}\rightarrow r;\}\\
\textsf{PRE-COND} \land\textsf{p}\in \textsf{NodeSet}(\textsf{root})\land\textsf{k}\in \textsf{Dom}(\textsf{p})
\end{array}\right)\label{SPEC-ASSIGN-3}
\end{equation}
$\textsf{k}\in \textsf{Dom}(\textsf{p})$ implies $\textsf{p}\neq \textbf{nil}$, thus $\textsf{k}<\textsf{p}\rightarrow K\lor \textsf{k}\ge\textsf{p}\rightarrow K$. From \textrm{IF-ST}, \ref{SPEC-ASSIGN-2}, \ref{SPEC-ASSIGN-3}, and :
\begin{equation}
\mathbb{P}\vdash\left(
\begin{array}{l}
\textsf{PRE-COND} \land \textsf{p}\in \textsf{NodeSet}(\textsf{root}) \land\textsf{k}\in \textsf{Dom}(\textsf{p}) \land \textsf{p}\rightarrow K\neq \textsf{k}\\
\ \ \ \ \ \{\mbox{\textbf{if} }(\textsf{k}<\textsf{p}\rightarrow K)\ \textsf{p}:=\textsf{p}\rightarrow l;\ \textbf{else}\ \textsf{p}:=\textsf{p}\rightarrow r;\}\\
\textsf{PRE-COND} \land \textsf{p}\in \textsf{NodeSet}(\textsf{root}) \land \textsf{k}\in \textsf{Dom}(\textsf{p})
\end{array}\right)\label{SPEC-IF}
\end{equation}
$\textsf{k}\in\textsf{Dom}(\textsf{p})$ implies that $\textsf{p}\neq \textbf{nil}$, thus $\textsf{p}\rightarrow K \neq \textsf{k}\lor\textsf{p}\rightarrow K = \textsf{k}$. From \textrm{WHILE-ST}, \ref{SPEC-IF}:
\begin{equation}
\mathbb{P}\vdash\left(\begin{array}{l}
\textsf{PRE-COND} \land\textsf{p}\in \textsf{NodeSet}(\textsf{root}) \land \textsf{k}\in \textsf{Dom}(\textsf{p})\\
\ \ \ \ \ \{\mbox{\emph{the while statement}}\}\\
\textsf{PRE-COND} \land\textsf{p}\in \textsf{NodeSet}(\textsf{root}) \land \textsf{k}\in \textsf{Dom}(\textsf{p})\land \textsf{p}\rightarrow K= \textsf{k}\\
\end{array}\right)\label{WHILE-ST-SPEC}
\end{equation}
From \ref{MAPMAPP} and the properties of finite map:
\begin{equation}
\begin{array}{l}
\mathbb{P}, \textsf{isHBST}(x), \textsf{p} \in \textsf{NodeSet}(x), \textsf{k}=\textsf{p}\rightarrow K\vdash \\
\ \ \ \ \ \textsf{TMapP}(x,\textsf{p})\dag\{\textsf{p}\rightarrow K\mapsto y\} = \textsf{Map}(x)\dag\{\textsf{k}\mapsto y\}
\end{array}\label{MAPEQ}
\end{equation}
From \ref{MAPEQ}, \textbf{substitution}:
\begin{equation}
\begin{array}{l}
\mathbb{P}, \textsf{PRE-COND}, \textsf{p} \in \textsf{NodeSet}(\textsf{root}), \textsf{k}=\textsf{p}\rightarrow K\vdash \\
\textsf{isHBST}(\textsf{root}) \land \textsf{MapP}(\textsf{root},\textsf{p})\dag\{\textsf{p}\rightarrow K \mapsto \textsf{d}\}=M \dag\{\textsf{k}\mapsto\textsf{d}\})
\end{array}\label{IMP}
\end{equation}
From the rule \textrm{ASSIGN-ST}, \ref{PD-ISOLATE} and $\&\textsf{p}\rightarrow D\not\in\{\&\textsf{root}, \&\textsf{p}, \&\textsf{p}\rightarrow K, \&\textsf{k}, \&\textsf{d}\}$:
\begin{equation}
\mathbb{P}\vdash\left(
\begin{array}{l}
(\textsf{isHBST}(\textsf{root})\land \textsf{MapP}(\textsf{root},\textsf{p})\dag\{\textsf{p}\rightarrow K \mapsto x\}=M \dag\{\textsf{k}\mapsto\textsf{d}\})[\textsf{d}/x]\\
\ \ \ \ \ \{\textsf{p}\rightarrow D := \textsf{d}\}\\
(\textsf{isHBST}(\textsf{root})\land \textsf{MapP}(\textsf{root},\textsf{p})\dag\{\textsf{p}\rightarrow K \mapsto x\}=M \dag\{\textsf{k}\mapsto\textsf{d}\})[\textsf{p}\rightarrow D/x]
\end{array}\right)\label{SPEC-ASSIGN-41}
\end{equation}
From the rule \textrm{CONSEQ}, \ref{IMP}, \ref{SPEC-ASSIGN-41}
\begin{equation}
\mathbb{P}\vdash\left(
\begin{array}{l}
\textsf{PRE-COND}\land\textsf{p}\in \textsf{NodeSet}(\textsf{root}) \land\textsf{k}\in \textsf{Dom}(\textsf{p})\land \textsf{p}\rightarrow K= \textsf{k}\\
\ \ \ \ \ \{\textsf{p}\rightarrow D := \textsf{d}\}\\
\textsf{isHBST}(\textsf{root})\land \textsf{Map}(\textsf{root}) =M \dag\{\textsf{k}\mapsto\textsf{d}\}
\end{array}\right)\label{SPEC-ASSIGN-4}
\end{equation}
From the rule \textrm{SEQ-ST}, \ref{SPEC-ASSIGN-1}, \ref{WHILE-ST-SPEC}, \ref{SPEC-ASSIGN-4}:
\begin{equation}
\mathbb{P}\vdash \textsf{PRE-COND} \ \ \{\textsl{Prog}\}\ \ \textsf{isHBST}(\textsf{root}) \land \textsf{Map}(\textsf{root})=M\dag\{\textsf{k}\mapsto \textsf{d}\}
\end{equation}

\section{Heuristics: virtual variables and pragmatic meaning of program statements}\label{SEC-HEU}
Generally speaking, a pointer program may create unbounded number of data objects during its execution.
These data objects usually interconnected through pointers. They are usually used to represent
abstract values which can be retrieved using recursively defined DRFs. We can view a set of interconnected data objects as a \emph{virtual variable}, which holds an abstract value retrieved using a DRF. Usually, such a data object set maintains a set of structural properties during the program execution. These properties can also be expressed using a set of boolean-typed DRFs. As we did in our running example, a DRF $\textsf{isHBST}$ is used to state that a set of data objects form a binary search tree, while the DRF $\textsf{Map}$ is used to retrieve a finite map from this binary search tree.

Usually, assigning new values to such a virtual variable is performed by a group of program statements. These statements change the values stored in a few number of the data objects, thus change the abstract value `stored' in the virtual variable. As to the structural properties, either none of the statements changes their values, or some statements change their values, but some other statements restore them afterwards. To reasoning the effect of these statements on the abstract value, we can define some auxiliary data-retrieve functions.
\begin{itemize}
\item These auxiliary DRFs do not accessed the memory units modified by these statements. So the abstract values retrieved by these auxiliary DRFs keep unchanged.
\item The relationship between the abstract value retrieved by the main DRFs, those retrieved by auxiliary DRFs, and the values stored in the modified memory units can be proved based on the definitions of the DRFs. For example, the property
$$
\mathbb{P}, \textsf{isHBST}(x), \textsf{y} \in \textsf{NodeSet}(x)\vdash \textsf{Map}(x)=\textsf{MapP}(x,y)\dag
\{y\rightarrow K\mapsto y\rightarrow D\}\label{MAPMAPP}
$$
shows the relation between the main DRF \textsf{Map}, the auxiliary DRF \textsf{MapP}, and the values stored in $\&y\rightarrow K$ and $\&y\rightarrow D$.
\item The values retrieved by auxiliary DRFs keep unchanged. The effect of these statements on the modified memory units can be relatively easily derived. So, the effect of these statements on the abstract value retrieved by main DRF can be reason based on the relations between main DRFs, auxiliary DRFs and the value stored in modified memory units.
\end{itemize}
To specify and verify these statements, we should understand and reason these statements as a whole, as these statements work together to assign a new value to a virtual variable. Understanding the effects of such statement groups can help us understand the whole program abstractly. In the appendix~\ref{APP-ABS}, we briefly describe such an example.
We say the effect of a group of program statements on a virtual variable as the \emph{pragmatic meaning} of these statements. Understanding and verifying the pragmatic meanings of small statement groups first, then we can verify code with larger size step by step.

\section{Conclusion and future works}\label{SEC-CONCLUSION}
In this paper, we present an extension of Hoare logic for verification of pointer programs. The pre-conditions and post-conditions are formulae of an extended version of the LPF logic, which can deal with undefinedness, recursive function definitions, and types. Program types and function symbols ($\ast$, $\&\!\!\rightarrow\!n$ and $\&[\,]$) associated with these types are introduced to model memory unit access and memory layout for composite types. A set of proof rules are introduced to specify these function symbols. Using these functions, people can deal with high-level program types (record, array) directly.

People can define recursive functions to retrieve abstract values from concrete interconnected data objects. We call these functions as data-retrieve functions (DRFs). Such functions can also be defined to specify the properties of data structures. For each data-retrieve function $f$, we can derive the definition of its corresponding memory-scope function (MSF) syntactically. When an abstract value is retrieved by applying $f$ to a set of arguments, applying the MSF of $f$ to same arguments results in a set of memory units accessed during the retrievement. As long as no memory unit in this set is modified during program executions, applying $f$ to same arguments results in same abstract value.

We present a new proof rule for assignment statements, and another rule for memory allocation statements. The proof rule for assignment statements says that after the assignment, the memory unit referred by the left-hand stores the value of the right-hand computed before the assignment. It also says that the abstract values keep unchanged if the memory unit referred by the left-hand is not in their memory scopes. The proof rule for memory allocation says that after the allocation, the memory unit referred by the left-hand stores a reference to a newly allocated memory block.

This logic has the following advantages.
\begin{itemize}
\item This logic is easy to learn. Most of the knowledge encoded in this logic have been (explicitly or implicitly) taught in undergraduate CS courses. For examples, the concept of recursive functions and first order logic are already taught in undergraduate CS courses. The proof rules about program variables, $\ast$, $\&\!\!\rightarrow\!n$, $\&[]$ are taught informally in the undergraduate courses about programming languages and compilers.
\item This logic supports reuse of proofs. Most of the proved properties of DRFs are about data structures. They are independent of the code under verification. So these properties can be reused in verification of other code using same data structures. It is possible to build a library of pre-defined DRFs, MSFs, and their properties.
\item Verification can be performed on different abstract levels. A group of statements change the abstract value represented by a set of interconnected data objects, but keep the structural properties of these data objects. People can first understand the \emph{pragmatic meaning} of these statements, i.e. the effect of these statements on the relevant abstract values. Then, they may view these data objects as a virtual variable, and the statements as an abstract statement assigning new value to this virtual variable. Thus, people can reasoning the program at a more abstract level.
\item Make use of the research results on pointer analysis. Many of the premises when applying proof rules can be proved automatically by pointer analysis. For example, for all assignment statements of the form $\ast(\&v) = e$, the premise that $\&v\neq \textbf{nil}$ can be proved by pointer analyer easily. For assignment statements of the form $\ast p=e$, the premise $\textsf{p}\neq \textbf{nil}$ of the proof rule ASSIGN-ST can also be verified automatically in many cases.
\end{itemize}

In the future, we will extended our logic to deal with more programming language concepts: function calls, function pointers, class/object, generics, $\dots$. At the mean time, we will try to build a library of pre-defined DRFs, MSFs, and their properties for frequently used data structures.

\appendix
\section{Another example: inserting a node to a binary search tree}\label{APP-BST-ADDNODE}

\begin{example}
The program depicted in Figure~\ref{PROG-2} add a new tuple $(\textsf{k},\textsf{d})$ into the map represented by a binary search tree.
The types of the program variables $\textsf{k}$ and
$\textsf{d}$ are both \textbf{integer}. The type of program variables $\textsf{rt}$ and $\textsf{tmp}$ are
\textbf{P}($T$), where $T$ is $\textbf{REC}((l,\textbf{P}(T))\times(r,\textbf{P}(T))\times
(K,\textbf{integer})\times(D,\textbf{integer}))$. The type of $\textsf{p}$ is $\textbf{P}(\textbf{P}(T))$.
\end{example}

\begin{figure}
\begin{center}
\parbox{300pt}{
\begin{tabbing}
\textsf{p}:=\&\textsf{rt};\\
\textbf{while} \= ($\ast \textsf{p} != \textbf{nil}$)\\
\{\\
\>\texttt{if} ($\textsf{k} < \textsf{p}\rightarrow K$ ) $\textsf{p} := \&(\ast \textsf{p})\rightarrow l$ \texttt{else} $\textsf{p} := \&(\ast \textsf{p})\rightarrow r$;\\
 \}\\
\textsf{tmp} = \texttt{alloc}($T$);\\
$\textsf{tmp}\rightarrow K := \textsf{k}$; $\textsf{tmp}\rightarrow D := \textsf{d}$;\\
$\ast\textsf{p}$=\textsf{tmp};
\end{tabbing}}
\end{center}
\caption{Another program}\label{PROG-2}
\end{figure}

The DRFs depicted in Figure~\ref{NEW-FUNC} are used in the specification and verification of the program depicted in Figure~\ref{PROG-2}.
If $\ast x$ points to the root node of a binary search tree, and $y$ is the address of a child-field of a node of this tree.
The DRF $\textsf{PNodeSet}(x,y)$ retrieve the set of the children-pointer-field addresses (i.e. addresses of the fields $l$ and $r$) of all the nodes in the binary search tree derived by setting $\ast y$ to \textbf{nil}. The argument $x$ is also in this set. $\textsf{MapPP}(x,y)$ retrieve the map represented by this modified binary search tree.

The boolean-typed DRF $\textsf{isHBSTK}(x,y)$ says that if we make $\ast y$ point to a newly allocated node $\{\textbf{nil}, \textbf{nil}, \textsf{k}, \textsf{d}\}$, $\ast x$ is still the root node of a binary search tree. The DRF $\textsf{DomK}(x,y)$ retrieve the keys stored in this tree.

The (simplified) definitions of the corresponding MSFs are depicted in Figure~\ref{NEW-MSFs}.
We use $\verb"DMK"_m$, $\verb"STK"_m$, $\verb"MPPP"_m$ as $\mathfrak{M}(\textsf{DomK})$, $\mathfrak{M}(\textsf{isSTK})$, and $\mathfrak{M}(\textsf{MapPP})$ respectively.
Let $\mathbb{P}'$ be the set of function definitions depicted in Figure~\ref{DATA-STRUCTURE-INTERPRETATION-FUNCTIONS} and Figure~\ref{NEW-FUNC}. Some of the properties about the DRFs in $\mathbb{P}'$  and corresponding MSFs are depicted in Figure~\ref{NEW-PROPER}.
Some of the DRFs and MSFs in Section~\ref{SEC-PROOF}, together with their properties, are reused in this verification.
\begin{figure}
\begin{tabbing}
$\textsf{DomK}(x,y): \textbf{P}(\textbf{P}(T))\times\textbf{P}(\textbf{P}(T))\rightarrow \textbf{SetOf}(\textbf{integer})$\\
\ \ \ \ $\triangleq$ \= $ (x= y)\ \ \ \ $         \= $\,?\ \{\textsf{k}\}:$\\
                     \> $(\ast x=\textbf{nil})$   \= $\,?\ \emptyset : \{(\ast x) \rightarrow K\} \cup \textsf{DomK}(\&(\ast x)\rightarrow l,y)\cup \textsf{DomK}(\&(\ast x)\rightarrow r,y)$\\
\\
$\textsf{isHBSTK}(x,y):\textbf{P}(\textbf{P}(T))\times\textbf{P}(\textbf{P}(T))\rightarrow \textbf{boolean}$\\
\ \ \ \ $\triangleq$ \> $(x=y)$ \>$\,?\ \texttt{TRUE}\,:$\\
                       \> $(\ast x=\textbf{nil})$ \>$\,?\ \texttt{TRUE} :$\\
                       \>     \>$\texttt{InHeap}(\ast x)\land \textsf{isHBSTK}(\&(\ast x)\rightarrow l,y)\land\textsf{isHBSTK}(\&(\ast x)\rightarrow r,y)\land$\\
                       \>     \>$(\textsf{DomK}(\&(\ast x)\rightarrow l)=\emptyset?\texttt{TRUE}: \texttt{MAX}(\textsf{DomK}(\&(\ast x)\rightarrow l))<(\ast x)\rightarrow K) \land$\\
                       \>     \>$(\textsf{DomK}(\&(\ast x)\rightarrow r)=\emptyset?\texttt{TRUE}:(\ast x)\rightarrow K < \texttt{MIN}(\textsf{DomK}(\&(\ast x)\rightarrow r)))$\\
\\
$\textsf{MapPP}(x,y):\textbf{P}(\textbf{P}(T))\times\textbf{P}(\textbf{P}(T))\rightarrow \textbf{Map integer to integer})$\\
\ \ \ \ $\triangleq$ \> $(x=y)$ \> $\,?\ \emptyset\,:$\\
                     \> $(\ast x=\textbf{nil})$ \> $\,?\,\emptyset\,:$\\
                     \>         \> $\textsf{MapPP}(\&(\ast x)\rightarrow l,y)\dag\textsf{MapPP}(\&(\ast x)\rightarrow r,y)\dag\{(\ast x)\rightarrow K\mapsto (\ast x)\rightarrow D\}$\\
\\
$\textsf{PNodeSet}(x,y) : \textbf{P}(\textbf{P}(T)) \times \textbf{P}(\textbf{P}(T)) \rightarrow \textbf{SetOf}(\textbf{Ptr})$\\
\ \ \ \ $\triangleq\{x\}\cup( (x=y)? \emptyset : (\ast x = \textbf{nil})\,?\,\emptyset : (\textsf{PNodeSet}(\&(\ast x)\rightarrow l,y)\cup \textsf{PNodeSet}(\&(\ast x)\rightarrow r,y))))$\\
\end{tabbing}
\caption{DRFs for specifying and verifying the program in Figure~\ref{PROG-2}}\label{NEW-FUNC}
\end{figure}

\begin{figure}
\begin{tabbing}
$\verb"DMK"_m(x,y):\textbf{P}(\textbf{P}(T))\times\textbf{P}(\textbf{P}(T))\rightarrow \textbf{SetOf}(\textbf{Ptr})$\\
\ \ \ \ $\triangleq$ \= $(x=y)\ \ \ \ \ $ \=\ $? \{\&\textsf{k}\} $\\
                     \> $(\{x\}\cup (\ast x=\textbf{nil}) ? \emptyset : \{x,\&(\ast x)\rightarrow K \}\cup \verb"DMK"_m(\&(\ast x)\rightarrow l,y)\cup \verb"DMK"_m(\&(\ast x)\rightarrow r,y))$\\
\\

$\verb"HBSTK"_m(x,y):\textbf{P}(\textbf{P}(T))\times\textbf{P}(\textbf{P}(T))\rightarrow \textbf{SetOf}(\textbf{Ptr})$\\
\ \ \ \ $\triangleq$ \> $(x=y)$ \>$\,?\ \emptyset\,:$\\
                     \> $\{x\}\cup (\ast x=\textbf{nil})\,?\ \emptyset : \verb"HBSTK"_m(\&(\ast x)\rightarrow l,y)\cup \verb"HBSTK"_m(\&(\ast x)\rightarrow r,y)\cup $\\
                       \>     \>$\verb"DMK"_m(\&(\ast x)\rightarrow l)\cup (\textsf{DomK}(\&(\ast x)\rightarrow l)=\emptyset?\emptyset : \{\&(\ast x)\rightarrow K\}) \cup$\\
                       \>     \>$\verb"DMK"_m(\&(\ast x)\rightarrow r)\cup (\textsf{DomK}(\&(\ast x)\rightarrow r)=\emptyset?\emptyset : \{\&(\ast x)\rightarrow K\})$\\
\\
$\verb"MPPP"_m(x,y):\textbf{P}(\textbf{P}(T))\times\textbf{P}(\textbf{P}(T))\rightarrow \textbf{SetOf}(\textbf{Ptr})$\\
\ \ \ \ $\triangleq$ \> $(x=y)\,?\emptyset\,:$\\
                     \> $\{x\}\cup (\ast x=\textbf{nil})\,?\,\emptyset\,:$\\
                     \>         \> $\verb"MPPP"_m(\&(\ast x)\rightarrow l,y)\cup\verb"MPPP"_m(\&(\ast x)\rightarrow r,y)\cup\{\&(\ast x)\rightarrow K, \&(\ast x)\rightarrow D\}$\\
\\
$\verb"PNS"_m(x,y) : \textbf{P}(\textbf{P}(T))\times \textbf{P}(\textbf{P}(T)) \rightarrow \textbf{SetOf}(\textbf{Ptr})$\\
\ \ \ \ $\triangleq (x=y)? \emptyset : \{x\}\cup (\ast x = \textbf{nil})\,?\,\emptyset : (\verb"PNS"_m(\&(\ast x)\rightarrow l,y)\cup \verb"PNS"_m(\&(\ast x)\rightarrow r,y))$\\

\end{tabbing}
\caption{MSFs of the DRFs in Figure~\ref{NEW-FUNC}}\label{NEW-MSFs}

\end{figure}

\begin{figure}
\begin{equation}
\mathbb{P}', \textsf{isHBST}(\ast x), \textsf{isHBSTK}(x,y) \vdash y\not\in \verb"DMK"_m(x,y)\cup\verb"HBSTK"_m(x,y)\cup \verb"MPPP"_m(x,y)\cup\verb"PNS"_m(x,y)\label{YIsolate}
\end{equation}

\begin{equation}
\mathbb{P}', \textsf{isHBST}(\ast x), y\in \textsf{PNodeSet}(x,y), \ast y\neq \textbf{nil} \vdash \ast y\in \textsf{NodeSet}(\ast x)
\end{equation}

\begin{equation}
\begin{array}{l}
\mathbb{P}', \textsf{isHBST}(\ast x), y\in \textsf{PNodeSet}(x,y)\land \textsf{isHBSTK}(x, y)\land \textsf{k}< (\ast y)\rightarrow K  \vdash\\
\ \ \ \ \ \ \&(\ast y)\rightarrow l \in\textsf{PNodeSet}(x,\&(\ast y)\rightarrow l) \land \textsf{isHBSTK}(x, \&(\ast y)\rightarrow l)
\end{array}
\end{equation}

\begin{equation}
\begin{array}{l}
\mathbb{P}', \textsf{isHBST}(\ast x), y\in \textsf{PNodeSet}(x,y)\land \textsf{isHBSTK}(x, y)\land \textsf{k}> (\ast y)\rightarrow K  \vdash\\
\ \ \ \ \ \ \&(\ast y)\rightarrow r \in\textsf{PNodeSet}(x,\&(\ast y)\rightarrow r) \land \textsf{isHBSTK}(x, \&(\ast y)\rightarrow r)
\end{array}\label{GO-R}
\end{equation}

\begin{equation}
\begin{array}{l}
\mathbb{P}', \textsf{isHBSTK}(x,y), y\in \textsf{PNodeSet}(x,y), \texttt{inHeap}(\ast y), \\
 (\ast y)\rightarrow K = \textsf{k} \land (\ast y) \rightarrow l = \textbf{nil} \land (\ast y)\rightarrow r = \textbf{nil}\end{array}
\vdash \textsf{isHBST}(\ast x)\label{HBST-COMB}
\end{equation}

\begin{equation}
\mathbb{P}', \textsf{isHBST}(\ast x), y\in \textsf{PNodeSet}(x,y) \vdash \textsf{Map}(\ast x) = \textsf{MapPP}(x,y)\dag \textsf{Map}(\ast y)\label{MAPPP-COMB}
\end{equation}

\begin{equation}
\mathbb{P}', \textsf{isHBST}(\ast x) \vdash \&\textsf{p}\not\in \verb"DMK"_m(x,y)\cup\verb"HBSTK"_m(x,y)\cup \verb"MPPP"_m(x,y)\label{PIsolate}
\end{equation}

\begin{equation}
\mathbb{P}', \textsf{isHBST}(\ast x) \vdash \&\textsf{tmp}\not\in \verb"DMK"_m(x,y)\cup\verb"HBSTK"_m(x,y)\cup \verb"MPPP"_m(x,y)\label{TMPIsolate}
\end{equation}

\caption{Some properties about the DRFs and MSFs}\label{NEW-PROPER}
\end{figure}

We use $\textsf{PRE-COND}$ as the abbreviation for $\textsf{isHBST}(\textsf{rt})\land \textsf{k}\not\in \textsf{Dom} (\textsf{rt}) \land \textsf{Map}(\textsf{rt})=M_0$.
The specification of this program is
$$\textsf{PRE-COND}  \{\texttt{The Program}\} \textsf{isHBST}(\textsf{rt})\land \textsf{Map}(\textsf{rt}) = M_0\dag \{\textsf{k}\mapsto \textsf{d}\}$$
The sketch of the proof is as follows. The common premise of these assertions is $\mathbb{P}'$, which is omitted for conciseness.\\
\\
From the rule ASSIGN-ST, \ref{P-ISOLATE}, \ref{PIsolate}, and $\&\textsf{p}\not\in\{\&\textsf{rt}, \&\textsf{k}\}$, we get following two assertions:
\begin{equation}\begin{array}{l}
(\textsf{PRE-COND}\land \textsf{isHBSTK}(\&\textsf{rt}, x)\land x\in \textsf{PNodeSet}(\&\textsf{rt},x))[\&\textsf{rt}/x]\\
\ \ \ \ \ \ \ \ \{\textsf{p}=\&\textsf{rt};\}\\
(\textsf{PRE-COND}\land \textsf{isHBSTK}(\&\textsf{rt}, x)\land x\in \textsf{PNodeSet}(\&\textsf{rt},x))[\textsf{p}/x]
\end{array}\label{PROG2-ST1}
\end{equation}
\\
\begin{equation}
\begin{array}{l}
(\textsf{PRE-COND} \land (x\in \textsf{PNodeSet}(\&\textsf{rt},x)) \land \textsf{isHBSTK}(\&\textsf{rt},x))[\&(\ast \textsf{p})\rightarrow r/x]\\
\ \ \ \ \ \ \ \{\textsf{p}=\&(\ast \textsf{p})\rightarrow r\}\\
(\textsf{PRE-COND} \land (x\in \textsf{PNodeSet}(\&\textsf{rt},x)) \land \textsf{isHBSTK}(\&\textsf{rt},x))[\textsf{p}/x]
\end{array}\label{GO-RIGHT}
\end{equation}
\\
From the rule CONSEQUENCE, \ref{GO-RIGHT}, and \ref{GO-R}:
\begin{equation}\begin{array}{l}
\textsf{PRE-COND} \land (\textsf{p}\in \textsf{PNodeSet}(\&\textsf{rt},\textsf{p}))\land \textsf{isHBSTK}(\&\textsf{rt},\textsf{p})\land
\ast\textsf{p}\neq\textbf{nil}\land \\(\textsf{k} > \ast y\rightarrow K)\\
\ \ \ \ \ \ \ \{\textsf{p}=\&(\ast \textsf{p})\rightarrow r\}\\
\textsf{PRE-COND} \land (\textsf{p}\in \textsf{PNodeSet}(\&\textsf{rt},\textsf{p})) \land \textsf{isHBSTK}(\&\textsf{rt},\textsf{p})
\end{array}\label{ELSE-BRANCH}
\end{equation}
\\
Similarly to the way we get \ref{ELSE-BRANCH}, we have:
\begin{equation}
\begin{array}{l}
\textsf{PRE-COND} \land (\textsf{p}\in \textsf{PNodeSet}(\&\textsf{rt},\textsf{p})) \land \textsf{isHBSTK}(\&\textsf{rt},\textsf{p}) \land \ast\textsf{p}\neq\textbf{nil} \land\\
(\textsf{k} < \ast y\rightarrow K)\\
\ \ \ \ \ \ \ \{\textsf{p}=\&(\ast \textsf{p})\rightarrow l\}\\
\textsf{PRE-COND} \land (\textsf{p}\in \textsf{PNodeSet}(\&\textsf{rt},\textsf{p})) \land \textsf{isHBSTK}(\&\textsf{rt},\textsf{p})
\end{array}\label{THEN-BRANCH}
\end{equation}
\\
As $\ast \textsf{p}\neq \textbf{nil}$ implies $\textsf{k}<\textsf{p}\rightarrow K\lor \textsf{k}\ge\textsf{p}\rightarrow K$. From the rule IF-ST, \ref{ELSE-BRANCH} and \ref{THEN-BRANCH}:
\begin{equation}
\begin{array}{l}
\textsf{PRE-COND} \land (\textsf{p}\in \textsf{PNodeSet}(\&\textsf{rt},\textsf{p})) \land \textsf{isHBSTK}(\&\textsf{rt},\textsf{p}) \land \ast\textsf{p}\neq\textbf{nil}\\
\ \ \ \ \ \ \ \{\texttt{if } (\textsf{k} < \textsf{p}\rightarrow K ) \textsf{ p} := \&(\ast \textsf{p})\rightarrow l \texttt{ else } \textsf{p} := \&(\ast \textsf{p})\rightarrow r;\}\\
\textsf{PRE-COND} \land (\textsf{p}\in \textsf{PNodeSet}(\&\textsf{rt},\textsf{p})) \land \textsf{isHBSTK}(\&\textsf{rt},\textsf{p})
\end{array}\label{IFST}
\end{equation}
\\
$\textsf{p}\in \textsf{PNodeSet}(\&\textsf{rt},\textsf{p}))$ implies $\textsf{p}\neq \textbf{nil}$, thus $\ast \textsf{p}=\textbf{nil} \lor \ast \textsf{p}\neq \textbf{nil}$, From the rule WHILE-ST, \ref{IFST}:
\begin{equation}
\begin{array}{l}
\textsf{PRE-COND} \land (\textsf{p}\in \textsf{PNodeSet}(\&\textsf{rt},\textsf{p})) \land \textsf{isHBSTK}(\&\textsf{rt},\textsf{p})\\
\ \ \ \ \ \ \ \{\texttt{the while statement} \}\\
\textsf{PRE-COND} \land (\textsf{p}\in \textsf{PNodeSet}(\&\textsf{rt},\textsf{p})) \land \textsf{isHBSTK}(\&\textsf{rt},\textsf{p})\land \ast \textsf{p}=\textbf{nil}
\end{array}
\end{equation}
\\
From the rule CONSEQUENCE, \ref{MAPPP-COMB}, $\ast\textsf{p}=\textbf{nil}$, and $\textsf{Map}(\textbf{nil})=\emptyset$, we have
\begin{equation}
\begin{array}{l}
\textsf{PRE-COND} \land (\textsf{p}\in \textsf{PNodeSet}(\&\textsf{rt},\textsf{p})) \land \textsf{isHBSTK}(\&\textsf{rt},\textsf{p})\\
\ \ \ \ \ \ \ \{\texttt{the while statement} \}\\
\textsf{p}\in \textsf{PNodeSet}(\&\textsf{rt},\textsf{p}) \land \textsf{isHBSTK}(\&\textsf{rt},\textsf{p})\land \textsf{MapPP}(\&\textsf{rt}, \textsf{p})=M_0
\end{array}\label{PROG2-WHILE-ST}
\end{equation}
\\
From the rule ALLOC-ST and the fact that $\textsf{tmp}$ is not relevant to any terms, we have:
\begin{equation}\begin{array}{l}
\textsf{p}\in \textsf{PNodeSet}(\&\textsf{rt},\textsf{p}) \land \textsf{isHBSTK}(\&\textsf{rt},\textsf{p})\land \textsf{MapPP}(\&\textsf{rt}, \textsf{p})=M_0\\
\ \ \ \ \ \ \ \ \{\textsf{tmp} = \texttt{alloc}(T);\}\\
\textsf{p}\in \textsf{PNodeSet}(\&\textsf{rt},\textsf{p}) \land \textsf{isHBSTK}(\&\textsf{rt},\textsf{p})\land \textsf{MapPP}(\&\textsf{rt}, \textsf{p})=M_0\land\\
\textsf{tmp}\neq \textbf{nil}\land \texttt{InHeap}(\textsf{tmp})\land \texttt{Unique}(\&\textsf{tmp}) \land \texttt{PtrInit}(\textsf{tmp})
\end{array}\label{PROG2-ALLOC-ST}\end{equation}
\\
\begin{equation}\begin{array}{l}
\textsf{p}\in \textsf{PNodeSet}(\&\textsf{rt},\textsf{p}) \land \textsf{isHBSTK}(\&\textsf{rt},\textsf{p})\land \textsf{MapPP}(\&\textsf{rt}, \textsf{p})=M_0\land\\
\textsf{tmp}\neq \textbf{nil}\land \texttt{InHeap}(\textsf{tmp})\land \texttt{Unique}(\&\textsf{tmp}) \land \texttt{PtrInit}(\textsf{tmp})\\
\ \ \ \ \ \ \ \ \{\textsf{tmp}\rightarrow K := \textsf{k}; \textsf{tmp}\rightarrow D:= \textsf{d};\}\\
\textsf{p}\in \textsf{PNodeSet}(\&\textsf{rt},\textsf{p}) \land \textsf{isHBSTK}(\&\textsf{rt},\textsf{p})\land \textsf{MapPP}(\&\textsf{rt}, \textsf{p})=M_0\land \\
\textsf{isHBST}(\textsf{tmp})\land \textsf{Map}(\textsf{tmp})=\{\textsf{k}\mapsto\textsf{d}\}\\
\end{array}\label{PROG2-FIELD-ASSIGN}
\end{equation}
\\
From the rule ASSIGN-ST, \ref{YIsolate}, and $\textsf{p}\not\in\{\&\textsf{p}\}$,  we have:
\begin{equation}\begin{array}{l}
\textsf{p}\in \textsf{PNodeSet}(\&\textsf{rt},\textsf{p}) \land \textsf{isHBSTK}(\&\textsf{rt},\textsf{p})\land \textsf{MapPP}(\&\textsf{rt}, \textsf{p})=M_0\land \\
\textsf{isHBST}(\textsf{tmp})\land \textsf{Map}(\textsf{tmp})=\{\textsf{k}\mapsto\textsf{d}\}\\
\ \ \ \ \ \ \ \ \{\ast \textsf{p}:=\textsf{tmp};\}\\
\textsf{p}\in \textsf{PNodeSet}(\&\textsf{rt},\textsf{p}) \land \textsf{isHBSTK}(\&\textsf{rt},\textsf{p})\land \textsf{MapPP}(\&\textsf{rt}, \textsf{p})=M_0\land \\
\textsf{isHBST}(\ast\textsf{p})\land \textsf{Map}(\ast\textsf{p})=\{\textsf{k}\mapsto\textsf{d}\}\\
\end{array}
\end{equation}
\\
From the rule CONSEQUENCE, \ref{HBST-COMB}, and \ref{MAPPP-COMB}, we have:
\begin{equation}\begin{array}{l}
\textsf{p}\in \textsf{PNodeSet}(\&\textsf{rt},\textsf{p}) \land \textsf{isHBSTK}(\&\textsf{rt},\textsf{p})\land \textsf{MapPP}(\&\textsf{rt}, \textsf{p})=M_0\land \\
\textsf{isHBST}(\textsf{tmp})\land \textsf{Map}(\textsf{tmp})=\{\textsf{k}\mapsto\textsf{d}\}\\
\ \ \ \ \ \ \ \ \{\ast \textsf{p}:=\textsf{tmp};\}\\
\textsf{isHBST}(\textsf{rt})\land \textsf{Map}(\textsf{rt})= M_0\dag \{\textsf{k}\mapsto\textsf{d}\}
\end{array}\label{PROG2-ST5}\end{equation}
\\
From the rule SEQ-ST, and \ref{PROG2-ST1}, \ref{PROG2-WHILE-ST}, \ref{PROG2-ALLOC-ST}, \ref{PROG2-FIELD-ASSIGN}, \ref{PROG2-ST5}, we prove the specification.
\begin{equation}
\textsf{PRE-COND}  \{\texttt{The Program}\} \textsf{isHBST}(\textsf{rt})\land \textsf{Map}(\textsf{rt}) = M\dag \{\textsf{k}\mapsto \textsf{d}\}
\end{equation}

\section{Verifying programs abstractly: the simplified Schorr-Waite algorithm}\label{APP-ABS}
The Schorr-Waite algorithm marks all nodes of a directed graph that are reachable form one given node. The program depicted in Figure~\ref{SCHORR-WAITE} is rewrite from a simplified version presented by David Gries\cite{SCHORRWAITE}.
The variables
$\textsf{tmp}, \textsf{p}, \textsf{q}, \textsf{root}, \textsf{vroot}$ are declared with type $\textbf{P}(T)$, and $T = \textbf{REC}((m,\textbf{integer})\times(l,\textbf{P}(T))\times(r,\textbf{P}(T))$. In this program, it is simplified
that each node has exactly two non-nil pointers (i.e. the field $l$ and $r$). We use this program to show how to verify a program in an abstract level.
The verification presented here is just a sketch, many details are omitted.

\begin{figure}
\begin{center}
\parbox{300pt}{
\begin{tabbing}
\textsf{p}=\textsf{root}; \textsf{q}=\textsf{vroot}; /*$\textsf{vroot}\rightarrow l = \textsf{vroot}\rightarrow r = \textsf{root}$*/\\
\textbf{while}\= ($\textsf{p}\neq \textsf{vroot}$)\=\\
\{\\
              \>$\textsf{p}\rightarrow m = \textsf{p} \rightarrow m + 1;$\\
              \>\textbf{if} ($\textsf{p}\rightarrow m = 3\textbf{ or }(\&\textsf{p}\rightarrow l)\rightarrow m = 0$)\\
              \>\{\\
              \>   \>$\textsf{tmp} := \textsf{p}; \textsf{p}:=\textsf{p}\rightarrow l;$\\
              \>   \>$\textsf{p}\rightarrow l=\textsf{p}\rightarrow r; \textsf{p}\rightarrow r := \textsf{q}; \textsf{q}=\textsf{tmp};$\\
              \>\}\\
              \>\textbf{else}\\
              \>\{\\
              \>   \>$\textsf{tmp} := \textsf{p}\rightarrow l;\textsf{p}\rightarrow l:= \textsf{p}\rightarrow r;$\\
              \>   \>$\textsf{p}\rightarrow r:= \textsf{q}; \textsf{q} := \textsf{tmp}$\\
              \>\}\\
\}
\end{tabbing}
}
\end{center}
\caption{The simplified Schorr-Waite algorithm}\label{SCHORR-WAITE}
\end{figure}

The DRFs used in (partial) specification and verification of this algorithm are depicted in Figure~\ref{SCHORR-FUN}. Intuitively speaking,
the DRF $\textsf{StackPath}(\textsf{p})$ retrieve the path from the virtual root $\textsf{vroot}$ to the current node $\textsf{p}$.
$\textsf{Pred}(x)$ is used to compute the predecessor of a node in the path. $\textsf{AcyclicSeq(x)}$ is used to assert that the path retrieved by
$\textsf{StackPath}(\textsf{p})$ is acyclic.

\begin{figure}
\begin{center}
\parbox{250pt}{
$\textsf{StackPath}(x):\textbf{P}(T)\rightarrow \textbf{SeqOf}(\textbf{P}(T))$\\
\mbox{}\ \ \ \ $\triangleq (x=\textsf{vroot})\,?\, [\textsf{vroot}] : [x]^\frown \textsf{StackPath}(\textsf{Pred}(x)))$\\

$\textsf{Pred}(x):\textbf{P}(T)\rightarrow \textbf{P}(T)$\\
\mbox{}\ \ \ \ $\triangleq (x\rightarrow m = 0)\,?\,\textsf{q} : ((x\rightarrow m = 1)\,?\,x\rightarrow r : x\rightarrow l)$\\

$\textsf{AcyclicSeq}(x):\textbf{SeqOf}(\textbf{P}(T))$\\
\mbox{}\ \ \ \ $\triangleq \textbf{head}(x)\not\in \textbf{tail}(x) \land \textsf{AcyclicSeq}(\textbf{tail}(x))$\\
\caption{The functions defined to prove Schorr-Waite algorithm}\label{SCHORR-FUN}
}
\end{center}
\end{figure}

Let $G$ be the node set of the graph; $L(\textsf{p})$ for original value of $\textsf{p}\rightarrow l$; $R(\textsf{p})$ for original value of $\textsf{p}\rightarrow r$; $\textsf{SUCC}(x)\triangleq (x\rightarrow m = 1)\,?\, R(x) : L(x)$. From \cite{SCHORRWAITE}, the following invariant of the while statement holds.
\begin{equation}
\begin{array}{l}
\forall x\in G\cdot (\ \ (x\rightarrow m = 0 \land x\rightarrow l = L(x)\land x\rightarrow r = R(x)) \lor\\
\mbox{\ \ \ \ \ \ \ \ \ \ \ \ \ \ \ }(x\rightarrow m = 1 \land x\rightarrow l = R(x)\land \textsf{SUCC}(x\rightarrow r) = x) \lor\\
\mbox{\ \ \ \ \ \ \ \ \ \ \ \ \ \ \ }(x\rightarrow m = 2 \land \textsf{SUCC}(x\rightarrow l) = x \land x\rightarrow r = L(x)) \lor\\
\mbox{\ \ \ \ \ \ \ \ \ \ \ \ \ \ \ }(x\rightarrow m = 3 \land x\rightarrow l = L(x) \land x\rightarrow r = R(x))\ \  )\\
\bigwedge (\ \ \ \\
\mbox{\ \ \ \ \ }(\textsf{p}\rightarrow m = 0 \land (L(\textsf{q})=\textsf{p}\lor R(\textsf{q})=\textsf{p})) \lor\\
\mbox{\ \ \ \ \ }(\textsf{p}\rightarrow m = 1 \land \textsf{q}=L(\textsf{p})) \lor (\textsf{p}\rightarrow m = 2 \land \textsf{q}=R(\textsf{p}))\ \ )\\
\bigwedge \textsf{AcyclicSeq}(\textsf{StackPath}(\textsf{p}))\land \textsf{p}=\texttt{head}(\textsf{StackPath}(\textsf{p}))
\end{array}\label{SCHOR-INV}
\end{equation}

We write this invariant as \textbf{INV}. The following specifications of the body of the while statement can be proved. In these specifications, $\stackrel\leftarrow p$ and $\stackrel\leftarrow S$ are constants used to denote the original value of $\textsf{p}$ and the path.
\begin{equation}
\begin{array}{l}
\textbf{INV}\land \textsf{p}\rightarrow m = 0 \land L(\textsf{p})\rightarrow m= 0 \land \textsf{StackPath}(\textsf{p})=\stackrel\leftarrow S\land \textsf{p}=\stackrel\leftarrow p\\
\mbox{}\ \ \ \ \mbox{\{The body of the while statement\}}\\
\textbf{INV} \land \stackrel\leftarrow p\rightarrow m = 1\land  \textsf{StackPath}(\textsf{p})= L(\stackrel\leftarrow p)\,^\frown \stackrel\leftarrow S
\end{array}
\end{equation}

\begin{equation}
\begin{array}{l}
\textbf{INV}\land \textsf{p}\rightarrow m = 1 \land R(\textsf{p})\rightarrow m= 0 \land \textsf{StackPath}(\textsf{p})=\stackrel\leftarrow S\land \textsf{p}=\stackrel\leftarrow p\\
\mbox{}\ \ \ \ \mbox{\{The body of the while statement\}}\\
\textbf{INV}\land \stackrel\leftarrow p\rightarrow m = 2\land  \textsf{StackPath}(\textsf{p})= R(\stackrel\leftarrow p)\,^\frown \stackrel\leftarrow S\\
\end{array}\end{equation}

\begin{equation}
\begin{array}{l}
\textbf{INV}\land \textsf{p}\rightarrow m = 2 \land \textsf{StackPath}(\textsf{p})=\textsf{p}^{\ \frown}\stackrel\leftarrow S\\
\mbox{}\ \ \ \mbox{ \{The body of the while statement\}}\\
\textbf{INV}\land \textsf{p}=R(\stackrel\leftarrow p)\land \stackrel\leftarrow p\rightarrow m = 3\land  \textsf{StackPath}(\textsf{p})=\stackrel\leftarrow S\\
\end{array}
\end{equation}
If we view $\textsf{StackPath}(\textsf{p})$ as a virtual variable, it can be seen that the body of the while statement have different pragmatic meanings when the value of $\textsf{p}\rightarrow m$ equals to $0,1,2$. Based on these properties, we can view the abstract program depicted in Figure~\ref{ABS-SCHORR-WAITE} as an abstract version of the program in Figure~\ref{SCHORR-WAITE}. From this abstract level, it is clear that the program in Figure~\ref{SCHORR-WAITE} is in fact an efficient and elaborative implementation of the depth-first-search algorithm. We can continue proving the algorithm based on this abstract program. Though assignment statements to abstract variables are not allowed in the code, the abstract program can help us thinking.
\begin{figure}

\begin{center}
\parbox{300pt}{
\begin{tabbing}
\textsf{p}:=\textsf{root}; \textsf{S}=$\emptyset$; $\textsf{push}(\textsf{vroot},\textsf{S})$; $\textsf{push}(\textsf{p},\textsf{S})$;\\
\textbf{while} \= ($\textsf{p}\neq \textsf{vroot}$)\= \textbf\ \ \textbf{do}\ \{\\
\>$\textsf{p}\rightarrow m = \textsf{p}\rightarrow m + 1$;\\
\>\textbf{if} ($\textsf{p}\rightarrow m  = 1 \land L(\textsf{p})\rightarrow m= 0$ )\ \ \ \ \ \{$\textsf{push}(L(\textsf{p}),\textsf{S});$ \}\\
\>\textbf{else}\ \=\textbf{if}($\textsf{p}\rightarrow m  = 2 \land R(\textsf{p})\rightarrow m= 0$ )\ \ \ \ \ \{$\textsf{push}(R(\textsf{p}),\textsf{S});$ \}\\
\>               \>\textbf{else}\ \=\textbf{if}($\textsf{p}\rightarrow m = 3$)\ \ \ \ \ \{$\textsf{pop}(\textsf{S});$\}\\
\>              \>              \>\textbf{else} \texttt{skip}\\
\>\textsf{p} = \textsf{top}(\textsf{S})\\
\>  \}
\end{tabbing}
}
\end{center}
\caption{The abstract version of the simplified Schorr-Waite algorithm}\label{ABS-SCHORR-WAITE}
\end{figure}

\end{document}